\documentclass[10pt,journal]{IEEEtran}

\usepackage{amsmath,amssymb,amsthm}
\usepackage{booktabs}
\usepackage{graphicx}
\usepackage{tikz}
\usetikzlibrary{arrows.meta,shapes.symbols}
\usepackage{enumitem}
\usepackage{cite}
\providecommand{\texorpdfstring}[2]{#1}

\newtheorem{theorem}{Theorem}
\newtheorem{lemma}[theorem]{Lemma}

\newtheorem{assumption}{Assumption}

\newcommand{\bB}{B}
\newcommand{\E}{\mathbb{E}}
\newcommand{\Prob}{\mathbb{P}}
\newcommand{\Nz}{\mathbb{N}_0}
\newcommand{\indic}[1]{\mathbf{1}\!\left[#1\right]}

\begin{document}

\title{A Queueing-Stability Criterion for Causal IPD-QIM Network Flow
Watermarking}

\author{Jiuxiang~Cao,~Guang~Cheng,~and~Guangjie~Liu%
\thanks{J. Cao is with the School of Cyber Science and Engineering, Southeast
University, Nanjing 211189, China, and also with Jiangsu Jinling Science
\& Technology Group Co., Ltd., Nanjing 210000, China
(e-mail: caojiuxiang@seu.edu.cn).}%
\thanks{G. Cheng is with the School of Cyber Science and Engineering, Southeast
University, Nanjing 211189, China (e-mail: chengguang@seu.edu.cn).}%
\thanks{G. Liu is with the School of Electronics and Information Engineering,
Nanjing University of Information Science and Technology, Nanjing 210044, China
(e-mail: gjieliu@nuist.edu.cn).}%
\thanks{Corresponding author: Guangjie Liu.}}

\markboth{IEEE TRANSACTIONS ON NETWORK AND SERVICE MANAGEMENT,~VOL.~XX, NO.~X, JULY~2026}%
{Cao \MakeLowercase{\textit{et al.}}: A Queueing-Stability Criterion for Causal IPD-QIM Flow Watermarking}

\IEEEtitleabstractindextext{%
\begin{abstract}
On multi-hop encrypted links such as Tor and cascaded VPNs, tunneling flattens
packet lengths and protocol fields, leaving inter-packet delay (IPD) as the main
carrier for active flow attribution. Causality, however, lets the embedder delay
packets but never advance them, so each quantization-index-modulation (QIM)
alignment injects nonnegative dwell into a delay buffer; unbounded dwell would
break lattice alignment through overflow and impose unacceptable delay on the
host connection. Whether a causal QIM watermark can be embedded stably on bursty
real traffic has largely been left to empirical configuration rather than
analysis. We model the embedder as a reflected dwell queue under the actual fixed
dual-lattice, equiprobable-random-bit rule, where the
injection is state-dependent---set by the current effective interval and the
bit---rather than an exogenous uniform variable. The substitution
$Y_i=\delta_i-r_i$ gives only an algebraic Lindley-form identity; stability is
governed by the busy-state drift at large dwell, where the effective interval
collapses to zero and the mean injection becomes $\Delta/4$. Hence, away from the
critical boundary, the buffer is stable if and only if $\mu_d>\Delta/4$ (i.e.\
$\Delta<4\mu_d$) for i.i.d.\ backgrounds, and, under stationary-ergodic and
finite-state Markov-modulated traffic with instantaneous overload, if and only if
the time-average intensity $\bar\rho<1$. With the exogenous decoding floor $\Delta\ge c\sigma_\xi$
($c=4Q^{-1}(\epsilon/2)$), this yields the operating window
$\Delta\in[c\sigma_\xi,4\bar\mu_d)$. Controlled simulations confirm a sharp
transition at $\rho=1$ set only by the mean; on four real application-level IPD
traces, with each simulated chain confined to a single flow (no cross-flow
splicing), the criterion gives the correct stability direction under flow-local
correlation and burstiness, while pooled cross-flow means overestimate the margin.
These results provide a testable stable-embeddability criterion and a
quantization-step configuration baseline for causal QIM network flow watermarking.
\end{abstract}

}

\maketitle
\IEEEdisplaynontitleabstractindextext
\IEEEpeerreviewmaketitle

\section{Introduction}\label{sec:intro}

\IEEEPARstart{N}{etwork} flow watermarking is an active traffic-analysis technique: at one
observation point the embedder perturbs the timing of a flow to embed a
recoverable identifier, so that the same flow can be re-identified elsewhere in
the network. It is a basic means of active traceback and attack attribution
on encrypted links~\cite{Wang2003Stepping} and remains an active research
direction~\cite{Li2026Survey}. On multi-hop encrypted paths such as Tor and
cascaded VPNs, tunnel encapsulation flattens packet lengths and protocol fields,
leaving inter-packet delay (IPD) as essentially the only usable watermark
carrier~\cite{Wang2007EffectiveAttack,houmansadr2009rainbow,houmansadr2011swirl}.
Quantization index modulation (QIM)~\cite{Chen2001QIM} quantizes the host to a
lattice that depends on the watermark bit and thus looks naturally suited to IPD
embedding. Flow watermarking, however, is causal: the embedder may only delay
packets, never advance them, whereas image or audio watermarking may freely
increase or decrease sample values. This asymmetry is not an implementation
detail but a prerequisite for whether the watermark can run stably.

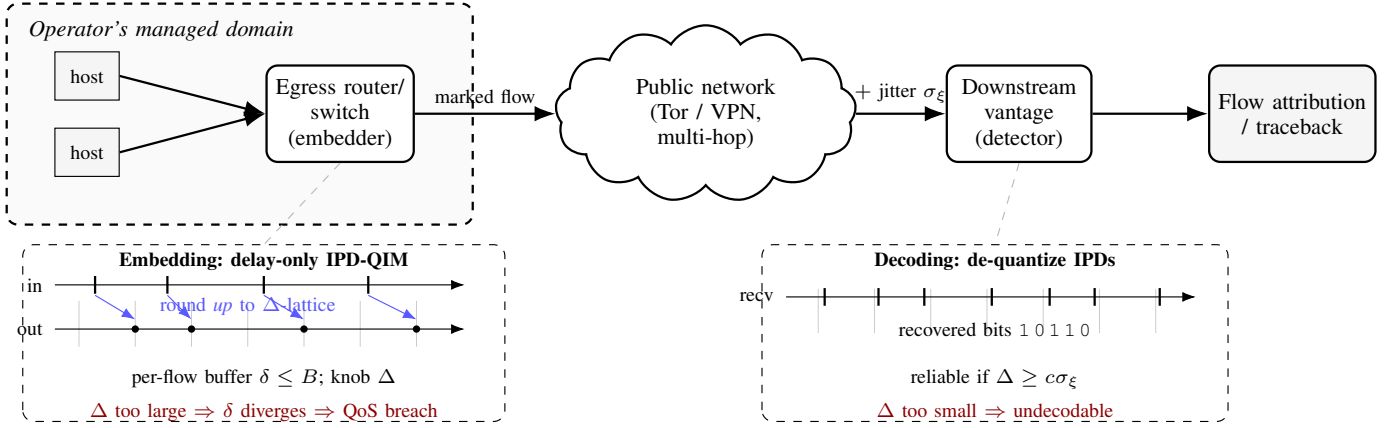
\begin{figure*}[t]
\centering
\resizebox{\textwidth}{!}{%
\begin{tikzpicture}[>=Latex,
  net/.style={draw, thick, rounded corners, align=center, font=\footnotesize, fill=white, minimum width=18mm, minimum height=12mm},
  host/.style={draw, align=center, font=\scriptsize, fill=black!4, minimum width=8mm, minimum height=6mm},
  flow/.style={-{Latex[length=2.6mm]}, thick},
  lead/.style={dashed, gray!55},
  fail/.style={font=\scriptsize, align=center, text=red!55!black}]

\draw[dashed, thick, rounded corners, fill=black!2] (-0.1,2.3) rectangle (5.7,5.05);
\node[font=\footnotesize\itshape, anchor=north west] at (0.05,4.97) {Operator's managed domain};
\node[host] (h1) at (0.9,4.15) {host};
\node[host] (h2) at (0.9,3.2) {host};
\node[net] (edge) at (4.05,3.68) {Egress router/\\switch\\(embedder)};
\draw[flow] (h1.east) -- (edge.west);
\draw[flow] (h2.east) -- (edge.west);
\node[cloud, draw, thick, cloud puffs=13, cloud puff arc=115, aspect=2.3, align=center,
      font=\footnotesize, minimum width=34mm, minimum height=22mm] (wan) at (8.6,3.68)
      {Public network\\(Tor / VPN,\\multi-hop)};
\draw[flow] (edge.east) -- node[above,font=\scriptsize]{marked flow} (wan.west);
\node[net] (det) at (12.5,3.68) {Downstream\\vantage\\(detector)};
\draw[flow] (wan.east) -- node[above,font=\scriptsize]{$+$ jitter $\sigma_\xi$} (det.west);
\node[net, fill=black!4] (attr) at (15.9,3.68) {Flow attribution\\/ traceback};
\draw[flow] (det.east) -- (attr.west);

\draw[lead] (edge.south) -- (3.1,2.05);
\draw[dashed, rounded corners] (0.1,-0.15) rectangle (6.1,2.05);
\node[font=\scriptsize\bfseries] at (3.1,1.85) {Embedding: delay-only IPD-QIM};
\foreach \x in {0.8,1.5,2.2,2.9,3.6,4.3,5.0} \draw[gray!40,very thin] (\x,0.78)--(\x,1.35);
\draw[->](0.5,1.55)--(5.6,1.55); \node[font=\scriptsize,anchor=east]at(0.47,1.55){in};
\foreach \x in {1.0,1.9,3.1,4.4} \draw[thick](\x,1.44)--(\x,1.66);
\draw[->](0.5,1.0)--(5.6,1.0); \node[font=\scriptsize,anchor=east]at(0.47,1.0){out};
\foreach \x in {1.5,2.2,3.6,5.0} \fill (\x,1.0) circle(1.3pt);
\foreach \a/\b in {1.0/1.5,1.9/2.2,3.1/3.6,4.4/5.0} \draw[->,blue!65](\a,1.42)--(\b,1.12);
\node[font=\scriptsize,blue!65]at(2.9,1.28){round \emph{up} to $\Delta$-lattice};
\node[fail,text=black,anchor=north]at(3.1,0.62){per-flow buffer $\delta\le B$;\ knob $\Delta$};
\node[fail,anchor=north]at(3.1,0.2){$\Delta$ too large $\Rightarrow$ $\delta$ diverges $\Rightarrow$ QoS breach};
\draw[lead] (det.south) -- (12.2,2.05);
\draw[dashed, rounded corners] (9.3,-0.15) rectangle (15.1,2.05);
\node[font=\scriptsize\bfseries] at (12.2,1.85){Decoding: de-quantize IPDs};
\foreach \x in {10.0,10.7,11.4,12.1,12.8,13.5,14.2} \draw[gray!40,very thin](\x,0.95)--(\x,1.5);
\draw[->](9.6,1.4)--(14.7,1.4); \node[font=\scriptsize,anchor=east]at(9.57,1.4){recv};
\foreach \x in {10.08,10.75,11.32,12.16,12.88,13.44,14.25} \draw[thick](\x,1.29)--(\x,1.51);
\node[font=\scriptsize]at(12.2,1.0){recovered bits \texttt{1\,0\,1\,1\,0}};
\node[fail,text=black,anchor=north]at(12.2,0.62){reliable if $\Delta\ge c\sigma_\xi$};
\node[fail,anchor=north]at(12.2,0.2){$\Delta$ too small $\Rightarrow$ undecodable};

\end{tikzpicture}%
}
\caption{Deployment view of causal IPD-QIM flow watermarking.}
\label{fig:deploy}
\end{figure*}

In operational terms, network flow watermarking is a capability the operator runs
at the boundary of the managed domain (Fig.~\ref{fig:deploy}): an edge router or
switch at the egress embeds the identifier by perturbing the timing of an outgoing
flow, and a detector at a downstream vantage point recovers it, so the operator can
attribute or trace that flow across the network it manages. The callouts in
Fig.~\ref{fig:deploy} zoom into the per-flow embedding and decoding mechanism
analyzed below. Because the embedder
acts on live service traffic at a forwarding device, deploying it is a
network-management task, and the quantization step $\Delta$ is effectively the only
knob the operator sets. Its two failure modes are both management failures: too
large a $\Delta$ lets the edge queue's delay grow without bound and breaks the
quality of service of the very traffic that carries the mark, whereas too small a
$\Delta$ makes the mark undecodable and renders the attribution capability useless.
Configuring the watermark is therefore the question of \emph{for what range of
$\Delta$ the edge embedder stays stable (QoS-safe) while remaining decodable
(attribution-effective)}.

The consequence of causality is structural. Each embedding aligns the current
packet to a lattice point no earlier than its arrival, injecting a nonnegative
alignment cost into a delay buffer, so the buffer depth accumulates packet by
packet. If this dynamics is unstable, the depth grows without bound with the
number of embedded packets: excessive dwell triggers budget truncation, so the
output IPD can no longer land exactly on the target lattice point, causing
overflow and decoding errors; the same dwell also acts on the real payload
packets, and once it exceeds transport- or application-layer tolerances it
forces connection resets or session interruption. The core question of causal
IPD-QIM is therefore not whether the embedder is convenient to implement, but
whether the watermark can be embedded stably without breaking the host traffic
or its own decodability. Classical QIM assumes the embedder may freely rewrite
host samples and hence never analyzes the stability of this
recursion~\cite{Chen2001QIM,Cox2007}; the flow-watermarking literature has long
used delay injection but mostly left its sustainability at the level of empirical
configuration.

We phrase this as the \emph{stable embeddability} of causal QIM watermarking:
given a background IPD process and a quantization step $\Delta$, does the
embedding dwell admit a finite steady state while the step still meets a basic
decoding-reliability requirement? Under a deployment-given buffer budget,
$\Delta$ is the main free parameter of the system, and its admissible range is
squeezed from two sides---queue stability sets an upper bound and decoding
reliability a lower bound. Working on the queue side, we answer three
progressive questions. (i)~Under what condition does the buffer recursion admit a
steady state with non-diverging dwell? (ii)~Does this stability survive bursty,
autocorrelated, even instantaneously overloaded real traffic? (iii)~Combining the
queue-side upper bound with the decoding-side lower bound, can one obtain a
testable operating window? The first two fix the upper bound
(Theorems~\ref{thm:stability} and~\ref{thm:robust}); the third adds the lower
bound and assembles the window (Section~\ref{sec:stab_thm}). The mathematical
object running through all three is the \emph{reflect-then-inject} dwell recursion
of causal QIM: a one-step substitution $Y_i=\delta_i-r_i$ writes it in standard
Lindley form (Lemma~\ref{lem:reduction}), but under a fixed lattice the injection
$r_i$ is coupled to the queue state and is not an exogenous independent input, so
stability is decided not by classical i.i.d.\ queueing results but by the drift
between the deep-buffer busy-state mean injection $\Delta/4$ and the background
service.

Our main contributions are as follows.
\begin{enumerate}[label=(\arabic*),leftmargin=2.2em,itemsep=2pt,topsep=2pt]
\item \textbf{Recursion identity and i.i.d.\ stability of the causal QIM buffer}
(Lemma~\ref{lem:reduction}, Theorem~\ref{thm:stability}). The dwell recursion is a
nonstandard \emph{reflect-then-inject} form; the substitution $Y_i=\delta_i-r_i$
writes it in Lindley form, but under a fixed lattice $r_i$ is set jointly by the
effective interval and the random bit and cannot be treated as an exogenous
uniform input. We therefore analyze the large-dwell drift: when the buffer is
deep the effective interval collapses to $0$, the random bit yields an average
injection of $\Delta/4$ over the two fixed lattices, and a Foster--Lyapunov
drift criterion gives the stability condition $\mu_d>\Delta/4$, translating
\emph{can the watermark be embedded stably} into the testable inequality
$\Delta<4\mu_d$.
\item \textbf{Burst-robust stability} (Theorem~\ref{thm:robust}). Under a
stationary-ergodic background (including finite-state Markov modulation), buffer
stability depends only on the time-average intensity ($\bar\rho<1$), not on the
instantaneous intensity: even if some intervals are instantaneously overloaded
($\rho(J)>1$, dwell accumulating), the queue is globally stable as long as the
slow intervals drain the buffer on average. This substantively extends the
i.i.d.\ condition $\mu_d>\Delta/4$ to real bursty traffic; on four
application-level real IPD traces, even the most autocorrelated scenarios keep a
finite steady-state dwell with the criterion pointing the right way.
\item \textbf{Decoding-side lower bound and the step-size operating window}
(Assumption~\ref{asm:decode}, Sections~\ref{sec:decoding} and~\ref{sec:stab_thm}).
From channel jitter, via a single-symbol error-rate upper bound, we derive a
sufficient lower bound $\Delta\ge c\sigma_\xi$ with $c=4Q^{-1}(\epsilon/2)$;
combined with the queue-stability upper bound it gives the conservative window
$\Delta\in[c\sigma_\xi,4\bar\mu_d)$, nonempty iff $\sigma_\xi<4\bar\mu_d/c$,
turning \emph{how to choose the step} from trial and error into an explicit
two-sided criterion. This lower bound and the resulting window supply the
constraint complementary to the stability upper bound.
\end{enumerate}

This work focuses on the stable embeddability of causal QIM watermarking in
encrypted-flow attribution: a stability criterion away from the critical
boundary, its robust extension to real bursty traffic, and the conservative
step-size operating window. Section~\ref{sec:related} reviews the QIM/causal
embedding and queueing-stability threads and summarizes the gap;
Section~\ref{sec:model} builds the queueing model, embedding rule, and basic
assumptions; Section~\ref{sec:stability} first derives the decoding-side
functional lower bound $\Delta\ge c\sigma_\xi$, then the recursion identity and
the i.i.d.\ / burst-robust stability results (the upper bound), and assembles the
window; Section~\ref{sec:numerical} validates the criterion on controlled
synthetic backgrounds and, chiefly, on real application-level IPD continuous
traces; Section~\ref{sec:conclusion} concludes.

\section{Related Work}\label{sec:related}

\subsection{Network flow watermarking}\label{sec:rel_flowwm}

Active flow watermarking perturbs a flow's timing or packet lengths at ingress
and detects the perturbation at egress to correlate the two ends; it is the main
route for active attribution of encrypted traffic, tracing back to Cabuk
et~al.'s~\cite{Cabuk2004} design and detection of IP covert timing channels. Wang
and Reeves~\cite{Wang2003Stepping} proposed a watermark-correlation scheme that
slightly adjusts the timing of selected packets, designed specifically to resist
timing perturbation, bringing IPD modulation into stepping-stone correlation;
Wang et~al.~\cite{Wang2007EffectiveAttack} further injected a unique watermark in
the inter-packet timing domain so that long flows are uniquely identifiable.
Later improvements followed two lines, lower visibility and higher robustness.
RAINBOW~\cite{houmansadr2009rainbow} is non-blind: it uses ingress-archived
original IPDs for differential detection, embedding with far smaller delay and
removing the self-interference of the host flow under blind detection;
SWIRL~\cite{houmansadr2011swirl} trades a per-flow marking pattern for
scalability while keeping delay small. The DSSS scheme of Yu
et~al.~\cite{Yu2007DSSS} instead uses direct-sequence spread spectrum, marginally
adjusting the sender rate with a pseudo-noise code to embed a covert
spread-spectrum signal. After 2018, schemes
diversified the carrier: Yang et~al.~\cite{Yang2022SlidingTor} embed with ON/OFF
timing on Tor and detect with a sliding-window $L_1$ distance,
HeteroTiC~\cite{Li2022HeteroTiC} combines order, timing, and length into a
heterogeneous channel, DynMark~\cite{Qiao2025DynMark} adds a dynamic
packet-count dimension, and Cui et~al.~\cite{Cui2023Hex} raise coding efficiency
with hexadecimal encoding. These schemes advance robustness or efficiency but all
treat delay injection as an implementation detail, and none analyzes whether the
injection can be sustained under the \emph{delay-only, never advance} constraint.

Alongside the embedding side, attacks and defenses have evolved. Kiyavash
et~al.~\cite{Kiyavash2008} showed that interval-based watermarks introduce
temporal correlation, enabling a multi-flow attack that, from a few marked flows,
detects the watermark, recovers secret parameters, and even removes it, and that
applies to both anonymous communication and stepping-stone detection. Lin and
Hopper~\cite{Lin2012} argued that being covert only to passive detection is
insufficient, and proposed stronger known/chosen-flow attack models. Recent
defenses such as DeMarking~\cite{Yuan2025DeMarking} rewrite IPDs in real time
with a generative adversarial network to erase the watermark while trying not to
disrupt the service. This pressure keeps pushing the embedding side toward
smaller, more controlled timing perturbations: the smaller the perturbation, the
more the embedder must align packets precisely to lattice points, and the more
critical the prerequisite of whether the buffer stays stable under the causal
constraint---precisely the concern of this paper.

At the survey level, Iacovazzi and Elovici~\cite{Iacovazzi2017Survey} classify
flow-watermarking algorithms by carrier, visibility, and robustness; the recent
problem-oriented survey of Li et~al.~\cite{Li2026Survey} reviews the 2001--2025
development and lists five open problems---identity representation,
embedding/detection placement, robustness to hopping, communication efficiency,
and fast detection in large background flows. That survey notes that ``embedding
should not markedly degrade the host's quality of service'' is a necessary
condition many methods overlook, yet still treats it as an efficiency--robustness
trade-off, without touching its root cause: whether the embedding delay grows
without bound under the \emph{delay-only, never advance} constraint.

\subsection{Quantization index modulation}\label{sec:rel_qim}

QIM, introduced by Chen and Wornell~\cite{Chen2001QIM}, is a \emph{provably good}
class of embedding: the sender picks one of a set of lattices (or a more general
codebook) according to the bit to embed and quantizes the host onto it; the
receiver only decides which lattice the (possibly noisy) signal is near, without
knowing the host itself. The same work proposes a distortion-compensated variant
(DC-QIM) that trades a linear combination of quantized and original signal for a
better noise-resistance/distortion balance. The theoretical justification of QIM
rests on Costa's \emph{writing on dirty paper}~\cite{Costa1983}: when the sender
\emph{non-causally} knows a Gaussian interference to be added on the channel, a
clever encoding (random binning / lattice coding) can fully cancel it, so the
capacity equals that of the interference-free channel. This non-causal assumption
is exactly what QIM-type methods rely on to approach that capacity: in image or
audio watermarking the host
is naturally known to the sender in advance, but in flow watermarking the arrival
times of future packets are precisely what the embedder cannot know ahead---this
mismatch is the starting point of the present paper. Cox
et~al.~\cite{Cox2007} give a widely used standard reference on the positioning of
QIM and spread-spectrum watermarking. QIM remains active: Mao
et~al.~\cite{Mao2024CAQIM} propose content-aware CA-QIM/CAMD-QIM to lower
distortion, and Lyu~\cite{Lyu2023Dithering} optimizes the dither structure to
improve PSNR/SSIM. All these assume the embedder may freely rewrite host samples
and need not consider the \emph{delay-only} physical constraint---natural for
image/audio but exactly what fails in the flow-watermarking setting above; the
two threads intersect at causality.

\subsection{Embedding queues and stability tools}\label{sec:rel_queue}

The Lindley recursion~\cite{Lindley1952} is the basis of steady-state analysis
for the single-server queue; Loynes~\cite{Loynes1962} proved, by a
backward-coupling construction, the necessary and sufficient condition (negative
drift) for the waiting-time sequence to converge to a unique stationary regime
under stationary-ergodic input; Foster~\cite{Foster1953} gave the drift criterion
for positive recurrence of Markov chains, and Meyn and
Tweedie~\cite{Meyn1993Markov} and Asmussen~\cite{Asmussen2003} developed the
stability theory of chains systematically. These tools target the standard
Lindley form \emph{add the increment, then reflect}, whereas the causal-QIM dwell
recursion is the nonstandard \emph{reflect, then inject} form
(Eq.~\eqref{eq:delta_update}), so classical results do not apply directly;
Lemma~\ref{lem:reduction} provides the reduction step bridging this gap. Notably,
Loynes' result requires no independence, only stationary ergodicity, and
Birkhoff's ergodic theorem~\cite{Durrett2019} together with the spectral theory
of Markov additive processes~\cite{Asmussen2003} are the companion tools needed to
extend stability to non-independent, bursty traffic. Empirically, Paxson and
Floyd~\cite{Paxson1995WideArea} observed long ago that wide-area IPDs deviate
markedly from Poisson and are bursty, so a stability analysis must tolerate
non-i.i.d., bursty real backgrounds.

\subsection{Feasibility of causal QIM}\label{sec:rel_gap}

Under the causal constraint, whether this embedding queue is feasible splits into
two progressive layers. The first, most basic: can the buffer grow without bound?
In classical QIM the embedder is assumed to rewrite host samples freely, so this
constraint never exists and the question never arises; the bipolar modulation of
spread-spectrum schemes runs into the \emph{cannot advance packets} difficulty
(Section~\ref{sec:why_qim} analyzes this for DSSS~\cite{Yu2007DSSS}-type schemes),
and existing work mostly circumvents it with per-packet positive bias or fixed
pre-buffering, rarely turning it into a testable stability criterion.
Theorem~\ref{thm:stability} in Section~\ref{sec:stability} is our answer to this
layer ($\mu_d>\Delta/4$). The second layer is harder: real traffic is neither
independent nor stationary---does the previous conclusion still hold? The key to
fixed-lattice random-bit QIM is not to reduce the injection to some uniform
distribution but to identify the deep-buffer busy-state mean injection $\Delta/4$;
on this drift structure we extend to Markov-modulated and stationary-ergodic
backgrounds and prove that stability is decided only by the time-average drift
(Theorem~\ref{thm:robust}). Two layers, two answers, together give the queue-side
upper bound on $\Delta$. Beyond the upper bound, $\Delta$ also needs a
decoding-side lower bound before a testable conservative window can be obtained;
this half comes from decoding reliability and is not a gap left by the two threads
above.

\section{A Queueing Model of the Causal QIM Embedder}\label{sec:model}

The causal QIM embedder is essentially a single-server queue: arriving packets
are customers, the quantization alignment cost is the service, and a packet's
dwell time is its waiting time. This section makes the correspondence
precise---first why, under causality, only QIM gives a queueable nonnegative
one-sided injection (Section~\ref{sec:why_qim}), then the dwell recursion from the
embedding rule (Section~\ref{sec:qim_embed}) and the budgeted buffer dynamics
(Section~\ref{sec:buffer}), and finally the assumptions the analysis relies on
(Section~\ref{sec:assump}). Table~\ref{tab:notation} collects the main symbols.

\begin{table}[!t]
\caption{Main symbols.}\label{tab:notation}
\centering\footnotesize
\begin{tabular}{@{}p{2.0cm}p{5.6cm}@{}}
\toprule
Symbol & Meaning \\
\midrule
$d_i$; $\mu_d,\bar\mu_d,\sigma_d$ & background IPD (inter-arrival), its mean, time-average mean, std \\
$\Delta$ & QIM quantization step \\
$r_i\in[0,\Delta)$ & lattice-alignment injection; under fixed lattices set by $\varepsilon_i,w_i$ \\
$\delta_i=t'_i-t_i$ & dwell of packet $i$ (buffer depth) \\
$Y_i=\delta_i-r_i$ & reduced waiting-time variable (Section~\ref{sec:reduction}) \\
$X_i=r_{i-1}-d_i$ & reduced increment; state-dependent under fixed lattices \\
$\rho=\Delta/(4\mu_d)$ & busy-state traffic intensity ($\bar\rho$: time-average) \\
$\bB$ & buffer budget \\
$\delta_\infty$ & steady-state dwell \\
$\sigma_\xi$ & std of channel net IPD jitter (exogenous) \\
$p(\gamma),\bar p(\gamma)$ & single-symbol error rate and its conservative bound, $\gamma=\Delta/\sigma_\xi$ \\
$\epsilon$ & target single-symbol error rate \\
$c=4Q^{-1}(\epsilon/2)$ & conservative step coefficient ensuring $p(\gamma)\le\epsilon$ \\
\bottomrule
\end{tabular}
\end{table}

\subsection{Causal feasibility of the modulation scheme}\label{sec:why_qim}

The choice of modulation scheme is tightly constrained by two premises:
\emph{one-sidedness}---causality dictates that the embedder may only delay a
packet, never advance it, so the injection $r_i\ge0$ always holds; and
\emph{statistical unpredictability}---the protocol of the watermarked traffic is
arbitrary, so the mean, variance, and even distributional shape of its IPDs are
unknown to the embedder and may switch abruptly within one session. Both
mainstream modulation ideas run into trouble under these constraints.

The first idea follows RAINBOW/SWIRL: encode a bit by aggregating an IPD
statistic over a symbol window---a slow window accumulates positive delay to raise
the mean, a fast window consumes pre-accumulated buffer to lower it. But the pair
is asymmetric: the slow window's \emph{spreading} can accumulate indefinitely,
whereas the fast window's \emph{shrinking} can only be supported by a pre-accumulated
buffer depth $D_0$; once a natural IPD exceeds $\mu_d+D_0$, that single late packet
drains the $D_0$-deep buffer at once, \emph{lowering the mean} fails, and no larger
buffer or stronger error correction can help.

The second idea is per-packet spread spectrum: superpose a pseudo-random bipolar
sequence of positive/negative adjustments on each IPD and recover the bit by
correlation. But a negative chip means the dwell must decrease, exactly what
causality forbids---the per-packet positive-bias remedy makes the accumulated
delay grow monotonically with the sequence length $N$, and the one-shot
pre-buffering remedy fails frequently in dense flows for lack of buffer.

Both failures point to one root cause: both encodings require the dwell to shrink
at certain moments, and shrinking is the one direction causality forbids. QIM
avoids this operation: its quantizer rounds upward ($q_w(x)\ge x$), each embedding
injecting only a nonnegative alignment cost $r_i\in[0,\Delta)$ into the buffer and
working from the zero state; encoding a $0$ or a $1$ both amount to rounding up to
a different lattice ($Q_0$ at $k\Delta$, $Q_1$ at $(k+\tfrac12)\Delta$), with no
forced draining. We stress, however, that \emph{each injection is nonnegative} is
only necessary, not sufficient, for stability---under the fixed-lattice random-bit
model without per-symbol dither, the feasibility criterion away from the critical
boundary, $\mu_d>\Delta/4$, is given in Section~\ref{sec:stability} by
Theorem~\ref{thm:stability}.

\subsection{The IPD-QIM embedding rule and dwell recursion}\label{sec:qim_embed}

Fix a step $\Delta>0$. To avoid mixing a physical-layer minimum inter-send gap
with the watermark lattice, we adopt a zero-phase convention: if an
implementation must keep a fixed minimum gap, first subtract it as a common
offset from the IPD, then analyze the extra dwell caused by the lattice. The two
complementary lattices are
\begin{equation}
Q_0=\{k\Delta\mid k\in\Nz\},\qquad
Q_1=\bigl\{(k+\tfrac12)\Delta\mid k\in\Nz\bigr\}.
\label{eq:qim_grids}
\end{equation}
With arrival time $t_i$ and send time $t'_i$ of packet $i$, embedding bit $w_i$
follows the causal QIM rule
\begin{equation}
t'_i = t'_{i-1} + q_{w_i}(\varepsilon_i),\qquad
\varepsilon_i=\max(0,\,t_i-t'_{i-1}),
\label{eq:qim_rule}
\end{equation}
where $q_w(x)\triangleq\min\{q\in Q_w\mid q\ge x\}$ rounds up to the nearest point
of $Q_w$ and $\varepsilon_i$ is the effective interval. Writing the dwell
$\delta_i\triangleq t'_i-t_i\ge0$, the background IPD $d_i\triangleq t_i-t_{i-1}$,
and the alignment injection $r_i\triangleq q_{w_i}(\varepsilon_i)-\varepsilon_i\in
[0,\Delta)$, we have $q_{w_i}(\varepsilon_i)=\varepsilon_i+r_i$. Substituting
into~\eqref{eq:qim_rule} gives $\delta_i=t'_{i-1}+\varepsilon_i+r_i-t_i$; with
$t'_{i-1}-t_i=\delta_{i-1}-d_i$ and splitting into two cases---empty
($d_i\ge\delta_{i-1}$, so $\varepsilon_i=t_i-t'_{i-1}$, giving $\delta_i=r_i$) and
busy ($d_i<\delta_{i-1}$, so $\varepsilon_i=0$, giving
$\delta_i=\delta_{i-1}-d_i+r_i$)---the merged dwell recursion is
\begin{equation}
\delta_i=\max(0,\,\delta_{i-1}-d_i)+r_i,\qquad r_i\in[0,\Delta).
\label{eq:delta_update}
\end{equation}
Under the fixed-lattice model $r_i$ is not an exogenous uniform variable but a
state-dependent injection set jointly by the effective interval $\varepsilon_i$
and the bit $w_i$. In particular, when the buffer is deep and $\varepsilon_i=0$,
$w_i=0$ gives $r_i=0$ and $w_i=1$ gives $r_i=\Delta/2$, so the busy-state mean
injection is $\Delta/4$. Decoding uses nearest-lattice hard decision: with
$\phi_i=d''_i\bmod\Delta$ ($d''_i$ the received IPD, $\phi_i$ the in-lattice
phase), $\hat{w}_i=\indic{\phi_i\in[\Delta/4,3\Delta/4)}$; with channel net IPD
jitter $\xi_i$ (std $\sigma_\xi$), the single-symbol hard-decision error rate has
the conservative bound $\Prob(|\xi_i|>\Delta/4)$. We focus on the embedder queue,
treat $\xi$ as exogenous, and only in Section~\ref{sec:decoding} use it to derive
the lower bound on the step; a full channel-side error analysis is out of scope.

\subsection{Buffer dynamics and the budget constraint}\label{sec:buffer}

The buffer's physical capacity is not infinite; its upper limit has two
independent sources. Section~\ref{sec:intro} noted the first: unbounded dwell
eventually hits the tolerance of the upper-layer service carried by the tunnel,
forcing connection resets or session interruption. The second lies in the
embedder itself: in real deployment it often maintains one such queue for each of
many concurrent flows (e.g.\ tens of Gbps at an enterprise egress), so a single
queue's memory grows linearly in the budget $\bB$ and the total memory grows with
the product of concurrent-flow count and $\bB$---a hardware constraint independent
of any single flow. We use $\bB$ (ms) to capture, in a unified way, the maximum tolerable dwell under
these two constraints. The
constrained recursion is
\begin{equation}
\delta_i=\min\!\Bigl(\max(0,\,\delta_{i-1}-d_i)+r_i,\;\bB\Bigr),\qquad
r_i\in[0,\Delta).
\label{eq:lindley}
\end{equation}
When the truncation $\min(\cdot,\bB)$ is active, the required dwell exceeds the
budget and the output IPD cannot land exactly on the target lattice point; we call
this an \emph{overflow}. We treat $\bB$ as a deployment-given fixed budget, using
it only to define overflow, not as a design variable. Two levels must be
distinguished: stability asks whether the unconstrained exact-alignment
recursion~\eqref{eq:delta_update} admits a finite steady state, decided only by
the large-dwell drift and independent of $\bB$; the budget asks, once stability
holds, how often a finite $\bB$ is breached. Once truncation occurs the output no
longer lands exactly on the target lattice and the subsequent effective interval
and phase change, so the truncated bounded chain cannot in turn prove stability.
Hence the analysis below runs on the unconstrained
recursion~\eqref{eq:delta_update}, with $\bB$ entering only as a threshold for the
downstream tail/overflow question. Define the effective traffic intensity
\begin{equation}
\rho\triangleq\frac{\Delta/4}{\E[d_i]}=\frac{\Delta}{4\mu_d},
\label{eq:rho}
\end{equation}
the ratio of mean injection to mean drain per packet in the deep-buffer busy
state. $\rho<1$ is the intuitive condition for the buffer not to grow without
bound, made rigorous in Theorem~\ref{thm:stability}.
Figure~\ref{fig:embedder_buffer} shows the per-packet processing flow, and
Figure~\ref{fig:buffer_dynamics} a typical dwell trajectory.

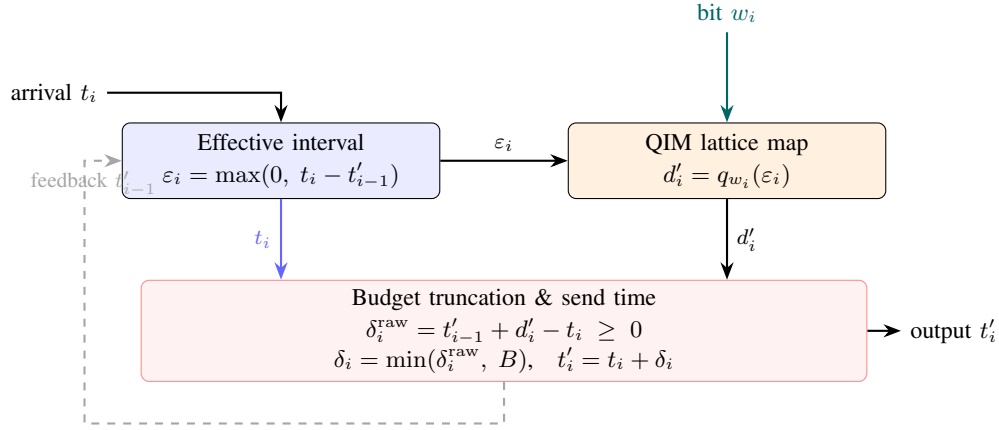
\begin{figure*}[t]
\centering
\begin{tikzpicture}[
  proc/.style={draw, rounded corners=3pt, minimum height=1.0cm, align=center, font=\small},
  arr/.style={->, >=Stealth, thick},
  lbl/.style={font=\footnotesize}
]
\node[proc, fill=blue!8, minimum width=4.2cm] (eff) at (0, 0)
    {Effective interval\\[2pt]$\varepsilon_i = \max(0,\; t_i - t'_{i-1})$};
\node[proc, fill=orange!12, minimum width=4.2cm] (qim) at (5.9, 0)
    {QIM lattice map\\[2pt]$d'_i = q_{w_i}(\varepsilon_i)$};
\node[font=\small] (in) at (-3.0, 0.9) {arrival $t_i$};
\draw[arr] (in.east) -| (eff.north);
\draw[arr] (eff) -- (qim) node[midway, above, lbl]{$\varepsilon_i$};
\draw[arr, teal!80!black] (5.9, 1.7) -- (qim.north) node[above, at start, font=\small]{bit $w_i$};
\node[proc, fill=red!5, draw=red!45, minimum width=9.6cm] (depart) at (2.95, -2.25)
    {Budget truncation \& send time\\[2pt]
     $\delta^{\mathrm{raw}}_i = t'_{i-1} + d'_i - t_i \;\ge\; 0$\\[1pt]
     $\delta_i = \min(\delta^{\mathrm{raw}}_i,\;\bB)$,\quad $t'_i = t_i + \delta_i$};
\draw[arr] (qim.south) -- (qim.south |- depart.north) node[midway, right, lbl]{$d'_i$};
\draw[arr, blue!60] (eff.south) -- (eff.south |- depart.north) node[midway, left, lbl]{$t_i$};
\node[font=\small] (out) at (8.9, -2.25) {output $t'_i$};
\draw[arr] (depart.east) -- (out);
\draw[arr, gray!70, dashed]
     (depart.south) -- ++(0, -0.55) -| (-2.6, 0) -- (eff.west)
     node[pos=0.25, below, lbl]{feedback $t'_{i-1}$};
\end{tikzpicture}
\caption{Per-packet processing flow of the causal QIM buffer.}\label{fig:embedder_buffer}
\end{figure*}

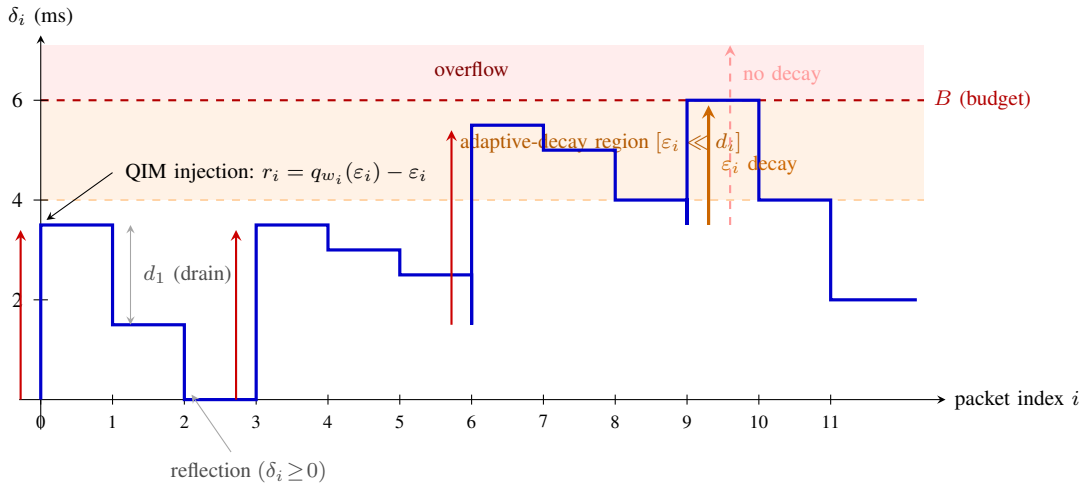
\begin{figure*}[t]
\centering
\begin{tikzpicture}[xscale=0.95,yscale=0.66,>=stealth,font=\footnotesize]
  \fill[red!7]    (0,6) rectangle (12.3,7.1);
  \fill[orange!10](0,4) rectangle (12.3,6);
  \draw[red!70!black,dashed,thick](0,6)--(12.3,6) node[right,font=\footnotesize]{$\bB$ (budget)};
  \node[red!50!black,font=\footnotesize] at (6,6.62){overflow};
  \draw[orange!60,dashed](0,4)--(12.3,4);
  \node[orange!70!black] at (7.8,5.15){adaptive-decay region\;$[\varepsilon_i \ll d_i]$};
  \draw[->](-0.3,0)--(12.6,0) node[right]{packet index $i$};
  \draw[->](0,-0.6)--(0,7.3) node[above]{$\delta_i$ (ms)};
  \foreach \y in {2,4,6}{\draw[thin](-0.1,\y)--(0.1,\y);\node[left] at (-0.12,\y){\y};}
  \foreach \x in {0,...,11}{\draw[thin](\x,-0.1)--(\x,0.1);\node[below,font=\scriptsize] at (\x,-0.12){\x};}
  \draw[blue!80!black,very thick]
    (0,0)--(0,3.5)--(1,3.5)--(1,1.5)--(2,1.5)--(2,0)--
    (3,0)--(3,3.5)--(4,3.5)--(4,3.0)--(5,3.0)--(5,2.5)--
    (6,2.5)--(6,1.5)--(6,5.5)--(7,5.5)--(7,5.0)--(8,5.0)--
    (8,4.0)--(9,4.0)--(9,3.5)--(9,6.0)--(10,6.0)--(10,4.0)--
    (11,4.0)--(11,2.0)--(12.2,2.0);
  \draw[->,red!80!black,thick](-0.28,0)--(-0.28,3.4);
  \draw[->,red!80!black,thick](2.72,0)--(2.72,3.4);
  \draw[->,red!80!black,thick](5.72,1.5)--(5.72,5.4);
  \draw[->,orange!80!black,very thick](9.3,3.5)--(9.3,5.9);
  \draw[->,red!40,dashed,thick](9.6,3.5)--(9.6,7.1);
  \node[orange!80!black,right] at (9.35,4.7){$\varepsilon_i$ decay};
  \node[red!40,right] at (9.65,6.5){no decay};
  \draw[<->,gray!70,thin](1.25,3.5)--(1.25,1.5);
  \node[gray!60!black,right] at (1.3,2.5){$d_1$ (drain)};
  \draw[->,thin,gray!70](2.65,-1.05)--(2.1,0.1);
  \node[gray!60!black,below] at (2.9,-1.05){reflection\;$(\delta_i\!\ge\!0)$};
  \node[right] at (1.05,4.55){QIM injection: $r_i = q_{w_i}(\varepsilon_i)-\varepsilon_i$};
  \draw[->,thin](1.0,4.55)--(0.08,3.6);
\end{tikzpicture}
\caption{A typical dwell trajectory ($\Delta=2$\,ms, budget $\bB=6$\,ms).}\label{fig:buffer_dynamics}
\end{figure*}

As Figure~\ref{fig:buffer_dynamics} shows, in each packet's processing the dwell
$\delta_i$ first drains naturally by the background interval $d_i$ (under the
reflection $\delta_i\ge0$) and is then raised by the QIM injection $r_i$; when the
dwell touches $\bB$ it enters the overflow region. The two region types depict the
adaptivity of the injection: when the dwell is shallow, the effective interval
$\varepsilon_i$ varies with arrivals and $r_i$ is small, so the buffer can fall
back via the background interval (adaptive-decay region); once the dwell is deep
enough that $\varepsilon_i\equiv0$, the injection no longer shrinks and the buffer
can only drain by $d_i$ (no-decay region)---exactly where overflow is prone and
the budget $\bB$ is needed.

\subsection{Basic assumptions}\label{sec:assump}

The analysis below relies on two assumptions on the queue dynamics
(Section~\ref{sec:stability}).

\begin{assumption}[Background IPD]\label{asm:bg}
The background inter-arrivals $\{d_i\}_{i\ge1}$ satisfy one of two conditions,
from weak to strong: in the baseline case they are i.i.d.\ with $d_i>0$ a.s.\ and
mean $\mu_d=\E[d_i]<\infty$; in a case closer to real traffic this is relaxed to
stationary-ergodic (including finite-state Markov modulation) with time-average
mean $\bar\mu_d<\infty$.
\end{assumption}

\begin{assumption}[Random bits and fixed lattices]\label{asm:inj}
The watermark bits $\{w_i\}_{i\ge1}$ are i.i.d.\ with
$\Prob(w_i=0)=\Prob(w_i=1)=1/2$ and independent of the background arrival process.
We introduce no per-symbol dither and do not assume a uniform effective phase;
the injection is always set by the fixed-lattice QIM rule
\[
r_i=R_{w_i}(\varepsilon_i),\qquad R_w(x)\triangleq q_w(x)-x .
\]
Hence $\{r_i\}$ is in general not i.i.d.\ but a state-dependent injection that
changes with the effective interval $\varepsilon_i=(d_i-\delta_{i-1})^+$.
\end{assumption}

Assumption~\ref{asm:bg} chooses a mathematical model for the arrival rhythm,
given in two tiers from baseline to generalization. The simplest tier treats each
inter-arrival as strictly positive, mutually independent, drawn from a common
distribution, with finite long-run mean---the most analyzable and easiest-to-
validate model, effectively memoryless: this instant's inter-arrival is unrelated
to the previous one. Real traffic is usually not so: sessions have busy and idle
phases, a quiet browsing spell may suddenly turn into a dense download, and
adjacent inter-arrivals are clearly correlated. The second tier is built for this:
it only requires the long-run average inter-arrival to be finite and the
statistics stable, allowing the traffic to switch among several states and
adjacent packets to be correlated, at the cost of replacing the simple
independent random walk in the proof by a more general ergodic theorem. The
conditions the stability conclusion relies on stop here---negative (time-average)
drift and ergodicity, with no further assumption on the shape of the background
distribution.

\section{Stable Embeddability and the Operating Window of
\texorpdfstring{$\Delta$}{Delta}}\label{sec:stability}

This section characterizes the stable embeddability of causal QIM watermarking and
the conservative admissible range of the step $\Delta$, squeezed from two sides and
given in turn. First a functional lower end (Section~\ref{sec:decoding}): to
guarantee decodability via a single-symbol sufficient bound, the step must meet a
conservative decoding floor $\Delta\ge c\sigma_\xi$ set by the channel jitter. Then
the main line: the algebraic identity between the causal-QIM recursion and the
Lindley form (Section~\ref{sec:reduction}), followed by a large-dwell drift
analysis for the fixed-lattice, random-bit state-dependent injection that gives the
i.i.d.\ stability condition $\mu_d>\Delta/4$ (Theorem~\ref{thm:stability}) and its
robust extension to bursty, autocorrelated traffic (Theorem~\ref{thm:robust}),
fixing the upper end $\Delta<4\bar\mu_d$; combining the two ends yields the
conservative, testable window $\Delta\in[c\sigma_\xi,4\bar\mu_d)$
(Section~\ref{sec:stab_thm}).

\subsection{Decoding feasibility: a functional lower bound on
\texorpdfstring{$\Delta$}{Delta}}\label{sec:decoding}

Before discussing queue stability, we establish a more basic functional
prerequisite: no matter whether the buffer is stable, if the decoder cannot
reliably recover the bit from the received IPDs the watermark does not stand. The
channel adds a propagation jitter $\xi$ to each IPD; the smaller $\Delta$ and the
denser the lattice, the more easily the nearest-lattice decision is flipped, so
\emph{decodability} itself forces a lower bound on the step. This bound is purely a
single-symbol SNR matter, independent of the decoding implementation---frame
synchronization, error-detecting codes, or multi-frame accumulation can further
lower the final error rate through observational redundancy but do not change the
single-symbol bound here.

\begin{assumption}[Decoding model]\label{asm:decode}
The channel adds i.i.d.\ jitter $\xi\sim\mathcal N(0,\sigma_\xi^2)$ to each IPD;
for the $\Delta$-step QIM nearest-lattice hard decision, the true single-symbol
error rate is $p(\gamma)$ with $\gamma=\Delta/\sigma_\xi$.
\end{assumption}

For any sent lattice point, if $|\xi|\le\Delta/4$ the received phase stays in the
correct decision region, so the error event is contained in $\{|\xi|>\Delta/4\}$.
Hence $p(\gamma)\le\bar p(\gamma)\triangleq2Q(\gamma/4)$ with
$Q(x)=\Prob(\mathcal N(0,1)>x)$. Here $\bar p$ is a conservative bound, not the
exact rate under periodic modulo decision, which would count as correct the cases
where noise shifts by whole steps back into a same-class lattice. Since $Q$ is
continuous and strictly decreasing, so is $\bar p(\gamma)$ on $\gamma>0$; thus for
any target $\epsilon\in(0,\tfrac12)$ there is a unique $c(\epsilon):=4Q^{-1}
(\epsilon/2)$ with $\bar p(c(\epsilon))=\epsilon$, and $\gamma\ge c(\epsilon)$
gives $p(\gamma)\le\epsilon$. This yields the functional lower bound
\begin{equation}
\Delta\ge c\,\sigma_\xi,\qquad c=4Q^{-1}\!\left(\tfrac{\epsilon}{2}\right),
\label{eq:decode_floor}
\end{equation}
i.e.\ the step must be at least $4Q^{-1}(\epsilon/2)$ times the jitter std to keep
the single-symbol error rate below $\epsilon$, depending only on the ratio
$\Delta/\sigma_\xi$.

This lower bound differs in origin from the stability upper bound below: it is a
decoding-side functional sufficient condition, not an intrinsic result of the
queue recursion. It rests on the idealized model of Assumption~\ref{asm:decode}
(i.i.d.\ Gaussian net jitter). $\xi$ is a channel perturbation distinct from the
host inter-arrival $d_i$; if the real $\xi$ is correlated or heavy-tailed, the
Gaussian tail underestimates extreme jitter, the $c$ needed for the same $\epsilon$
grows, and it must be recalibrated to the deployed channel's measured jitter. Thus
$\Delta\ge c\sigma_\xi$ gives a conservative order-of-magnitude lower end that
shifts as the channel model is refined but does not change the queue-intrinsic
upper end or the stability conclusion.

\subsection{Recursion identity and state-dependent injection}\label{sec:reduction}

The recursion~\eqref{eq:delta_update} is \emph{reflect, then inject} (first
$\max(0,\delta_{i-1}-d_i)$, then add $r_i$), unlike the standard Lindley form
\emph{add the increment, then reflect} for single-server waiting times. The
substitution below still writes it in Lindley form, but this is only an algebraic
identity: under a fixed lattice $r_i$ is set by $\varepsilon_i$ and $w_i$ and thus
coupled to the queue state, so one cannot treat $r_i$ as an exogenous uniform
input and directly invoke the classical Loynes result.

\begin{lemma}[Recursion identity]\label{lem:reduction}
For the budget-free dwell recursion~\eqref{eq:delta_update}, let
$Y_i\triangleq\delta_i-r_i=\max(0,\delta_{i-1}-d_i)\ge0$. Then for any realization
of $\{(d_i,r_i)\}$---without any independence or stationarity assumption---the
standard Lindley recursion
\begin{equation}
Y_i=\bigl(Y_{i-1}+X_i\bigr)^+,\qquad X_i\triangleq r_{i-1}-d_i,
\label{eq:lindley_std}
\end{equation}
holds, with $(x)^+=\max(0,x)$ and dwell $\delta_i=Y_i+r_i$. Under a fixed lattice
$X_i$ is in general not an i.i.d.\ increment, since
$r_{i-1}=R_{w_{i-1}}((d_{i-1}-\delta_{i-2})^+)$ depends on the past queue state;
hence Lemma~\ref{lem:reduction} states only the recursion structure and does not
by itself give a stability condition.
\end{lemma}

\begin{proof}
By the definition of $Y_i$, $\delta_{i-1}=Y_{i-1}+r_{i-1}$. Substituting into
$Y_i=\max(0,\delta_{i-1}-d_i)$ gives
$Y_i=\max(0,\,Y_{i-1}+r_{i-1}-d_i)=(Y_{i-1}+X_i)^+$, i.e.~\eqref{eq:lindley_std}.
The dwell $\delta_i=Y_i+r_i$ follows from the definition. The proof uses neither
the independence nor the distribution of $r_i$.
\end{proof}

The engineering meaning: the causal QIM buffer is still a queue, with the
background IPD $d_i$ providing natural drain and the alignment cost $r_i$ a new
workload; but under a fixed lattice this workload is not an exogenous service
demand---it adapts to the effective interval. When the buffer is shallow,
$\varepsilon_i=(d_i-\delta_{i-1})^+$ inherits the background phase and $r_i$ varies
with it; when the buffer is deep, almost all packets have $\varepsilon_i=0$, the
system enters the busy state, and the average injection degenerates to
$\E_w[R_w(0)]=\Delta/4$. Thus stability is decided not by a global average phase
injection but by the net drift in the large-dwell region.

\subsection{Stability: i.i.d.\ criterion and burst robustness}\label{sec:stab_thm}

\begin{theorem}[Stability of fixed-lattice random-bit QIM]\label{thm:stability}
Under Assumptions~\ref{asm:bg} and~\ref{asm:inj}, with the background arrival
distribution non-degenerate and satisfying the usual Markov-chain small-set
condition, if
\[
\mu_d>\frac{\Delta}{4}\qquad(\rho<1),
\]
then the budget-free dwell chain~\eqref{eq:delta_update} is positive recurrent,
admits a finite unique stationary distribution $\delta_\infty$, and converges to
it from any finite initial value. If the overflow event of packet $i$ is defined
as \emph{the unconstrained exact-alignment dwell exceeds the deployment budget
$\bB$,} the steady-state tail probability is $\Prob(\delta_\infty>\bB)$; this tail
characterizes budget risk and does not enter the stability threshold itself.
Conversely, if $\mu_d<\Delta/4$, the average drift in the large-dwell region is
positive, the chain has no finite stationary distribution, and any finite budget
is breached persistently. The critical boundary $\mu_d=\Delta/4$ depends on finer
distributional structure and is not taken as a deployable operating point.
\end{theorem}

\begin{proof}
Take $V(x)=x$. Given $\delta_{i-1}=x$, $\varepsilon_i=(d_i-x)^+$ and
$\delta_i=x-\min(d_i,x)+R_{w_i}((d_i-x)^+)$, so the one-step drift is
\[
\begin{aligned}
D(x)&=\E[\delta_i-x\mid\delta_{i-1}=x]\\
&=-\E[\min(d_i,x)]+\E\!\left[R_{w_i}\bigl((d_i-x)^+\bigr)\right].
\end{aligned}
\]
Since $0\le R_w(\cdot)<\Delta$ and $(d_i-x)^+\to0$ a.s.\ as $x\to\infty$, dominated
convergence gives $\lim_{x\to\infty}D(x)=-\mu_d+\E_w[R_w(0)]$. In the zero-phase
fixed dual lattice, $R_0(0)=0$ and $R_1(0)=\Delta/2$, so
$\E_w[R_w(0)]=\Delta/4$. If $\mu_d>\Delta/4$, there exist $x_0,\eta>0$ with
$D(x)\le-\eta$ for $x\ge x_0$; the Foster--Lyapunov drift
criterion~\cite{Foster1953,Meyn1993Markov} gives positive recurrence and a finite
stationary distribution. If $\mu_d<\Delta/4$, the large-dwell drift is positive,
and the standard reverse-drift criterion shows no finite stationary distribution
exists.

The budget part only interprets the steady-state tail: after the unconstrained
model reaches steady state, $\Prob(\delta_\infty>\bB)$ gives the probability that
the required dwell exceeds the budget. If the real system uses hard
truncation~\eqref{eq:lindley}, truncation changes the subsequent state and phase,
so its long-run truncation frequency is a tail problem of the truncated
implementation and cannot be used to change or prove the above threshold.
\end{proof}

Theorem~\ref{thm:stability} translates \emph{is embedding feasible} into a testable
engineering inequality: as long as the background mean inter-arrival exceeds the
busy-state mean injection ($\mu_d>\Delta/4$), the buffer admits a stationary
distribution and the dwell does not diverge. The contrapositive is equally strong:
when $\rho>1$ ($\Delta>4\mu_d$) the large-dwell drift is positive, the buffer grows
without bound, and QIM is unusable. The stability condition thus sets an upper
bound $\Delta<4\mu_d$. This bound comes from the mean injection $\Delta/4$ at
$\varepsilon_i=0$ under the real lattice rule, not the $\Delta/2$ of a
uniform-phase assumption; if the implementation adds per-symbol dither,
non-equiprobable coding, dynamic phase, or feedback step, the busy-state mean
injection changes with the rule and the upper bound must be recomputed from the
new drift.

Combining this stability upper end with the decoding lower end
$\Delta\ge c\sigma_\xi$ of Section~\ref{sec:decoding}, the conservative admissible
range of the step is squeezed from both sides,
\begin{equation}
\Delta\in[\,c\sigma_\xi,\ 4\mu_d\,),\qquad c=4Q^{-1}(\epsilon/2).
\label{eq:window}
\end{equation}
The lower end $c\sigma_\xi$ is the decodability floor (keeping the single-symbol
error rate below $\epsilon$), the upper end $4\mu_d$ the queue-stability ceiling of
fixed-lattice random-bit QIM; both depend only on measurable deployment
quantities---channel jitter $\sigma_\xi$ and background inter-arrival $\mu_d$. The
window is nonempty iff $\sigma_\xi<4\mu_d/c$; when jitter is too large or the
background too dense so it is empty ($\sigma_\xi\ge4\mu_d/c$), a fixed step cannot
meet both ends, and the problem should turn to online estimation, adaptive step,
or feedback control rather than tuning a single fixed $\Delta$.

The above assumes i.i.d.\ background IPDs; real traffic is bursty and
autocorrelated---SSH switching from interactive to bulk transfer, or the first
packet after a video re-buffering, makes the local inter-arrival drop sharply, so
the instantaneous intensity $\rho(J)=\Delta/(4\mu_d(J))$ exceeds $1$ in some
intervals. Whether stability survives such input is exactly its robustness. Since
the large-dwell drift of fixed-lattice QIM depends only on the busy injection
$\Delta/4$ and the time-average background drain, the same criterion extends to
stationary-ergodic input.

\begin{theorem}[Robust stability]\label{thm:robust}
Let the background $\{d_i\}$ be modulated by an irreducible aperiodic finite-state
Markov chain $\{J_i\}$, with $\E[d_i\mid J_i=j]=\mu_d(j)<\infty$; the watermark
bits satisfy Assumption~\ref{asm:inj}, and the joint chain $(\delta_i,J_i)$
satisfies the same irreducibility and small-set conditions as in
Theorem~\ref{thm:stability}. Write $\bar\mu_d=\E_\pi[d]=\sum_j\pi_j\mu_d(j)$ ($\pi$
the stationary distribution) and $\bar\rho=\Delta/(4\bar\mu_d)$. If $\bar\rho<1$
(equivalently $\bar\mu_d>\Delta/4$), the budget-free dwell chain admits a finite
stationary distribution. In particular, even if some modulation states are
instantaneously overloaded ($\rho(J_i)=\Delta/(4\mu_d(J_i))>1$, dwell
accumulating), the buffer is globally stable as long as $\bar\rho<1$; if
$\bar\rho>1$, the large-dwell average drift is positive and the system is
unstable.
\end{theorem}

\begin{proof}
Consider the joint chain $(\delta_i,J_i)$. Write the one-step expected drain
$m(j)\triangleq\E[d_i\mid J_{i-1}=j]=\sum_k P_{jk}\mu_d(k)$ ($P$ the transition
matrix); by $\pi P=\pi$, $\sum_j\pi_j m(j)=\bar\mu_d$. Given
$(\delta_{i-1},J_{i-1})=(x,j)$, as in Theorem~\ref{thm:stability}
$\delta_i-x=-\min(d_i,x)+R_{w_i}((d_i-x)^+)$, so the one-step drift
\[
\begin{aligned}
D(x,j)={}&-\E[\min(d_i,x)\mid j]\\
&+\E\!\left[R_{w_i}\bigl((d_i-x)^+\bigr)\,\middle|\,j\right]
\end{aligned}
\]
satisfies $\lim_{x\to\infty}D(x,j)=\Delta/4-m(j)$ (drain by monotone convergence to
$m(j)$, injection by dominated convergence to $\E_w[R_w(0)]=\Delta/4$). In an
overloaded state ($m(j)<\Delta/4$) this limit is positive, so the one-step drift of
$V=\delta$ is not uniformly negative at large $\delta$ and Foster--Lyapunov does
not apply directly---the technical crux of tolerating instantaneous overload.

To remove the state-dependence of the drain, introduce a correction. Since
$\{J_i\}$ is finite, irreducible, aperiodic and $\bar\mu_d-m$ has zero $\pi$-mean,
the Poisson equation $g(j)-\sum_k P_{jk}g(k)=\bar\mu_d-m(j)$ has a bounded solution
$g$ (on the zero-$\pi$-mean subspace where $I-P$ is invertible). Take
$V(\delta,j)=\delta+g(j)$; its one-step drift is
\[
\begin{aligned}
\E\bigl[V(\delta_i,J_i)&-V(\delta_{i-1},J_{i-1})\mid x,j\bigr]\\
&=D(x,j)+\Bigl(\textstyle\sum_k P_{jk}g(k)-g(j)\Bigr)\\
&=D(x,j)+\bigl(m(j)-\bar\mu_d\bigr),
\end{aligned}
\]
using the Poisson equation. Letting $x\to\infty$, the state-dependent $m(j)$ is
cancelled by the drift of $g$:
\[
\begin{aligned}
\lim_{x\to\infty}\E[\Delta V\mid x,j]
&=\bigl(\tfrac{\Delta}{4}-m(j)\bigr)+\bigl(m(j)-\bar\mu_d\bigr)\\
&=\tfrac{\Delta}{4}-\bar\mu_d,
\end{aligned}
\]
independent of $j$. If $\bar\mu_d>\Delta/4$, there exist $x_0,\eta>0$ with this
drift $\le-\eta$ for all $j$ and $\delta\ge x_0$; since $g$ is bounded,
$\{(\delta,j):\delta\le x_0\}$ is a small set, and the Foster--Lyapunov
criterion~\cite{Meyn1993Markov} gives positive recurrence of the joint chain and a
finite stationary distribution of $\delta_i$. The correction $g(J)$ uses the
modulation's own fluctuation to cancel the deviation of the drain $m(J)$ around
$\bar\mu_d$, so overloaded states are compensated by slow states in the
time-average sense---the rigorous source of \emph{stability decided by the
time-average drift $\Delta/4-\bar\mu_d$.} Conversely if $\bar\mu_d<\Delta/4$, the
limit is positive, $V$ has uniformly positive drift at large $\delta$, and the
reverse-drift criterion shows no finite stationary distribution.
\end{proof}

Theorem~\ref{thm:robust} shows buffer stability is decided by the time-average
intensity, not the instantaneous one. Even with whole intervals of instantaneous
overload ($\rho(J)>1$, dwell accumulating there), the queue does not diverge as
long as the slow intervals drain the buffer on average ($\bar\rho<1$). This
extends the i.i.d.\ condition $\mu_d>\Delta/4$ to $\bar\mu_d>\Delta/4$ for
correlated, bursty traffic. Correspondingly the window's upper end generalizes
from $4\mu_d$ to the time-average $4\bar\mu_d$, i.e.\
$\Delta\in[c\sigma_\xi,4\bar\mu_d)$; the form is unchanged, only the ceiling now
set by the time-average inter-arrival. Section~\ref{sec:numerical} first checks
the sharpness of this transition on synthetic backgrounds, then examines the
position of $\Delta$ relative to the empirical ceiling $4\hat\mu_d$ on real
application-level IPDs, observing whether the steady dwell on continuous segments
stays finite.

\section{Experiments}\label{sec:numerical}

The aim here is not a full flow-watermarking system evaluation but to test how the
stability criterion of Section~\ref{sec:stability} behaves in two kinds of
evidence. Section~\ref{sec:num_synthetic} is the synthetic-data experiment: on
controlled backgrounds we check, one by one, the recursion identity, the
phase-transition critical point and its shape-independence, the near-critical
scaling, burst-robust stability, and the decoding floor and step-size window,
confirming that the numerical implementation matches the theoretical mechanism.
Section~\ref{sec:num_realtraffic} is the real-data experiment and the main body of
this section: without assuming i.i.d.\ or finite-state Markov traffic, we take
pcap-extracted application-level IPD continuous traces as background, first
checking in the safe region whether the stability margin is positive and the
steady dwell bounded, then sweeping the step along the real background to approach
and cross the stability boundary $\rho=1$, to check that the criterion holds in the
safe region and correctly locates the divergence critical point.

All queue-side experiments run on the budget-free dwell recursion
(Eq.~\eqref{eq:delta_update}) with no budget cap, because stability only asks
whether the unconstrained exact-alignment dwell diverges: estimating the steady
dwell $\E[\delta_\infty]$ directly on this recursion cleanly separates a bounded
steady state from unbounded growth, and $\bB$ is only a threshold for overflow
risk after stability holds, not part of the stability-threshold estimate. The
design parameter under study is the step $\Delta$, whose admissible range is
squeezed by the stability upper bound and the decoding lower bound.

\subsection{Synthetic-data experiments}\label{sec:num_synthetic}

These experiments place the theoretical objects of Section~\ref{sec:stability}
under controlled conditions and check them one by one: confirming the numerical
implementation and the quantitative form of each conclusion. The checks use a
lognormal IPD as baseline background (the distribution-family sweep, two-state
Markov modulation, and non-ideal-jitter checks vary the background or jitter as
needed): lognormal has positive support and a right-skewed heavy tail, the queue
recursion solves stably on it, and it suits a systematic sweep of the criterion's
parameter region. What is needed is only a controllable distribution and stable
computation, not that lognormal be the unique model of real wide-area traffic;
real IPDs are often fitted by lognormal, Weibull, or gamma heavy-tailed
distributions~\cite{Arfeen2019Weibull}, and below we verify that the divergence
critical point is independent of the distribution family, so this choice does not
sway the criterion. The baseline takes mean $\mu_d=20$\,ms and coefficient of
variation $\mathrm{CV}=0.5$ (std $\sigma_d=10$\,ms); watermark bits are drawn
i.i.d.\ equiprobably, and the injection is computed by the fixed dual-lattice QIM
rule~\eqref{eq:qim_rule} from the effective interval $\varepsilon_i$ and bit
$w_i$, with no per-symbol dither and $r_i$ not preset to a uniform distribution.
Each check varies the relevant parameter---CV, intensity $\rho=\Delta/(4\mu_d)$,
two-state Markov modulation, and channel jitter $\sigma_\xi$---with values given in
place. Steady quantities are estimated from $2\times10^3$--$8\times10^3$
independent chains of $5\times10^3$--$1.2\times10^4$ packets each (dropping a
warm-up), lengthening chains near the critical point; full parameters are in the
released scripts.

These controlled experiments only confirm the numerical implementation and the
quantitative form of each theoretical conclusion, and are not an engineering
validation independent of the real-data evidence of
Section~\ref{sec:num_realtraffic}. The per-theorem checks are collected in
Figure~\ref{fig:synthetic}, the window and non-ideal-jitter checks in
Figure~\ref{fig:window}.

\subsubsection{Recursion identity and the i.i.d.\ stability transition}

On a stable configuration $\Delta=24$\,ms ($\rho=\Delta/(4\mu_d)=0.30$) we
simulate the fixed-lattice random-bit recursion and extract the reduced sequence
by $Y_i=\delta_i-r_i$: the steady samples give
$\E[\delta_\infty]-\E[Y_\infty]=10.71$\,ms, matching the empirical mean injection
$\E[r_i]=10.71$\,ms and confirming $\delta_i=Y_i+r_i$ as a path-wise identity.
Meanwhile the sample correlation of $Y_i$ and $r_i$ is $-0.306$, showing that
under the fixed lattice the injection is strongly coupled to the queue state and
cannot be written as a convolution of an independent $Y$ and an independent
uniform $r$. On the busy subsample ($\varepsilon_i=0$): the busy fraction is about
$0.277$ and its empirical mean injection is $6.00$\,ms, exactly the theoretical
$\Delta/4=6$\,ms. This confirms the key fact of the derivation: the reduction is
only a recursion identity, and the stability threshold is set by the deep-buffer
busy-state mean injection $\Delta/4$.

Figure~\ref{fig:synthetic}(a) sweeps $\rho=\Delta/(4\mu_d)$ over four same-mean
($\mu_d=20$\,ms) backgrounds from different families---lognormal (CV $=0.5$,
$1.0$), Weibull ($k=1.2$), gamma ($k=2$): the four curves are uniformly bounded
for $\rho<1$, and the divergence knee lands sharply at $\rho=1$ regardless of
family or fine shape, a direct manifestation of Theorem~\ref{thm:stability}'s
critical point set only by the mean. The dwell before the knee rises with a
heavier tail (at $\rho=0.9$, from $132$\,ms for lognormal CV$0.5$ to $194$\,ms for
CV$1.0$), but the divergence position does not move. Thus using a lognormal
background does not sway the conclusion: real IPDs are often heavy-tailed and
markedly non-Poisson~\cite{Paxson1995WideArea}, fitted by lognormal, Weibull, or
gamma~\cite{Arfeen2019Weibull}, and the criterion sees only the mean; the
real-data evidence (Section~\ref{sec:num_realtraffic}) uses measured IPDs
directly, free of any parametric distribution.

Figure~\ref{fig:synthetic}(b) depicts the divergence rate near the upper bound:
the steady dwell $\E[\delta_\infty]$ is approximately linear in $1/(1-\rho)$
(linear fit $R^2=0.999$), i.e.\ it diverges as the reciprocal of the margin
$1-\rho$. This gives the quantitative cost near the bound---each step of $\Delta$
toward $4\bar\mu_d$ scales the steady dwell by $1/(1-\rho)$; so a real deployment
should not set $\Delta$ near the pooled boundary but keep a margin commensurate
with traffic heterogeneity before $\rho=1$.

\subsubsection{Burst-robust stability}

Figure~\ref{fig:synthetic}(c) uses a two-state Markov background (slow state
$\mu_A=36$, burst state $\mu_B=4$\,ms, symmetric transitions making the
time-average $\bar\mu_d=20$\,ms) and sweeps the time-average intensity
$\bar\rho=\Delta/(4\bar\mu_d)$: although the burst-state instantaneous intensity
$\rho_B$ reaches $4.75\times$ overload, as long as $\bar\rho<1$ the steady dwell is
bounded ($\E[\delta_\infty]=553$\,ms at $\bar\rho=0.95$), and the knee still lands
exactly at $\bar\rho=1$, not $\rho_B=1$. In particular, $\Delta=24$\,ms
($\bar\rho=0.30$, burst state $\rho_B=1.5$ overloaded, half the time in that state)
gives $\E[\delta_\infty]=17.6$\,ms---half the time instantaneously overloaded yet
the queue stable, exactly Theorem~\ref{thm:robust}'s \emph{stability decided by the
time-average drift.} From the proof, this half-time overload corresponds to a
positive one-step drift $\Delta/4-m(J)>0$ in the burst state (here
$\Delta/4=6>\mu_B=4$), which the correction $g(J)$ of Theorem~\ref{thm:robust}
averages away over the modulation cycle, so the joint chain's effective drift
converges to the state-independent $\Delta/4-\bar\mu_d<0$; the knee locked at
$\bar\rho=1$ rather than $\rho_B=1$ is the numerical evidence of this cancellation.

To generalize \emph{the time-average is the only decider} from a single
configuration, fix $\bar\rho=0.80$ and vary only the burst structure: mild burst
($\rho_B=1$), strong burst ($\rho_B=6$), and strong-persistent burst ($\rho_B=6$,
stay probability $0.9$) give steady dwells $13.2$, $27.0$, $62.8$\,ms
respectively, all bounded. The stronger and more persistent the burst, the larger
the dwell, but stability itself is unaffected---as long as the time-average drift
is negative the queue is stable, the burst structure changing only the magnitude
of the dwell, not whether it diverges.

\subsubsection{Decoding floor and the step-size operating window}

Figure~\ref{fig:synthetic}(d) checks the decoding-side sufficient lower bound.
Under Gaussian jitter we Monte-Carlo the single-symbol error rate $p(\gamma)$ of
the QIM nearest-lattice hard decision ($\gamma=\Delta/\sigma_\xi$) against the
sufficient bound $2Q(\gamma/4)$: for $\gamma\gtrsim5$ the two nearly coincide (the
measured-to-bound ratio is near $1$, at most $1.03$, within Monte-Carlo error),
so the sufficient bound is nearly tight in the working region and not a
conservative overestimate; only for $\gamma\lesssim4$ (noise comparable to the
step, modulo wraparound significant) is the measurement clearly below the bound.
At the design threshold $c(\epsilon)=4Q^{-1}(\epsilon/2)$ the measured error rate
matches the target $\epsilon$ (at $c(10^{-2})=10.30$ and $c(10^{-3})=13.16$ the
measured $p/\epsilon$ ratios are $1.004$ and $1.005$, within Monte-Carlo error),
confirming that the floor $\Delta\ge c\sigma_\xi$ is both sufficient and tight
under the ideal channel model.

\begin{figure*}[t]
\centering
\includegraphics[width=0.98\textwidth]{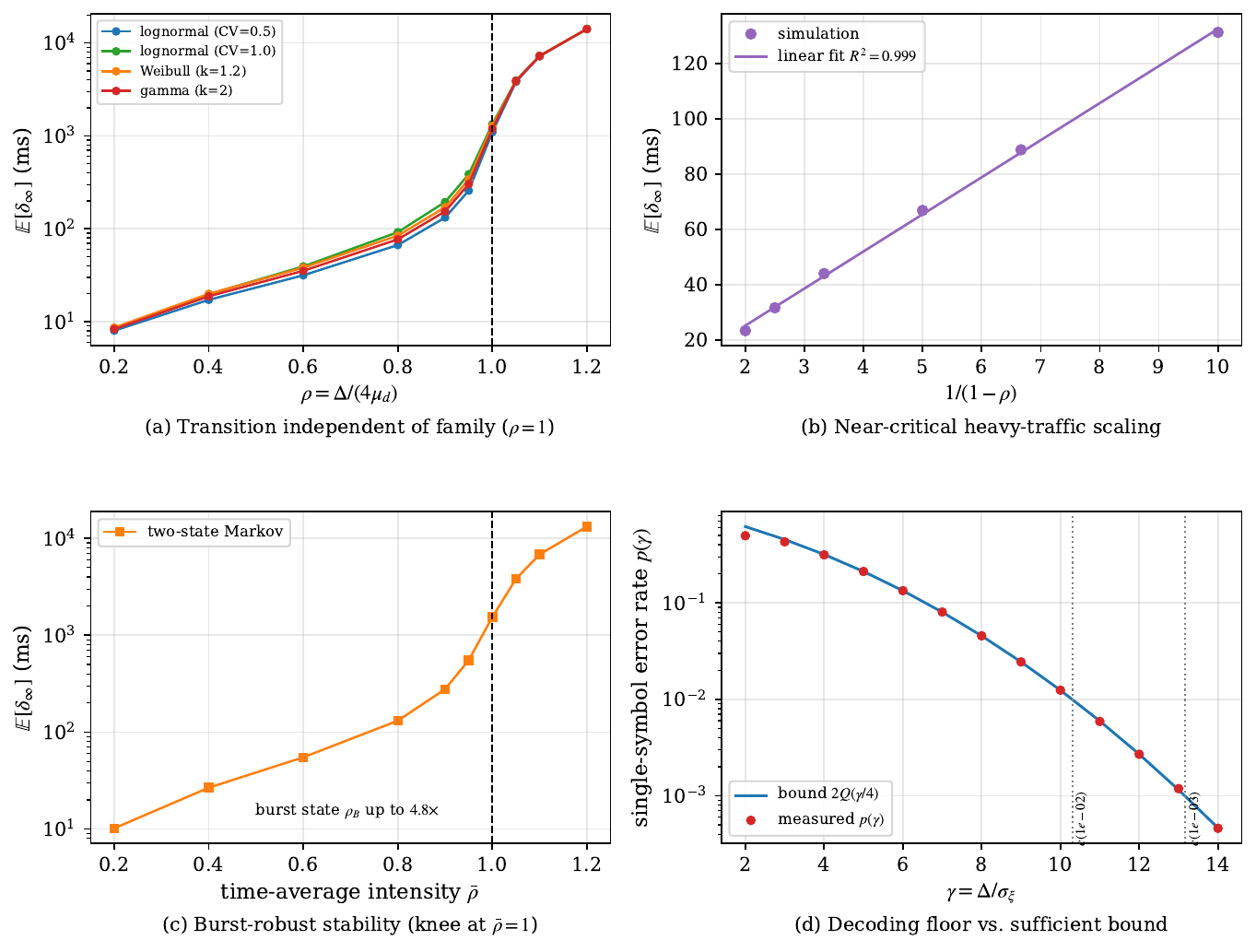}
\caption{Per-theorem fine-grained checks on synthetic backgrounds.}
\label{fig:synthetic}
\end{figure*}

Once the stability upper end and decoding lower end are separately verified,
Figure~\ref{fig:window}(a) puts both on the same $\Delta$ axis and validates the
window $[c\sigma_\xi,4\mu_d)$ as a whole. With $\mu_d=20$\,ms, jitter
$\sigma_\xi=2$\,ms, target $\epsilon=10^{-2}$ ($c=10.3$, so lower end
$c\sigma_\xi=20.6$\,ms, upper end $4\mu_d=80$\,ms), we estimate along $\Delta$ both
the single-symbol error rate and the steady dwell. Three regions are clear: for
$\Delta<20.6$\,ms the error rate exceeds $\epsilon$ ($\Delta=12$ gives $p=0.13$),
undecodable; for $\Delta\ge80$\,ms the steady dwell enters the critical or
divergent region ($\Delta=80$ already at $1166$\,ms, $\Delta=84$ rising to
$4280$\,ms and exploding further); only inside $[20.6,80)$\,ms (e.g.\ $\Delta=24$,
$p=2.7\times10^{-3}$, $\E[\delta_\infty]=12.3$\,ms; $\Delta=40$,
$p=5.7\times10^{-7}$, $\E[\delta_\infty]=23.4$\,ms) is the watermark both decodable
and the queue stable. This two-sided squeeze is the complete characterization of
stable embeddability along the step; when jitter grows so that
$c\sigma_\xi\ge4\mu_d$ the window is empty and a fixed step can no longer meet both
constraints.

The lower end $\Delta\ge c\sigma_\xi$ rests on i.i.d.\ Gaussian jitter; heavy-tailed
real jitter makes the floor optimistic. Figure~\ref{fig:window}(b) confirms this:
under variance-matched Gaussian, Laplace, and Student-$t_3$ (heavy-tailed) jitter
we measure the single-symbol error rate. The Gaussian case matches the theoretical
$c$ ($10.31$ vs.\ $10.30$ at $\epsilon=10^{-2}$); but heavy-tailed jitter needs a
markedly larger $c$ for the same $\epsilon$---at $\epsilon=10^{-2}$, Laplace and
$t_3$ rise to $13.0$ and $13.3$ (about $26\%$--$29\%$ over Gaussian), at
$\epsilon=10^{-3}$ Laplace rises to $19.6$ ($49\%$ over), and $t_3$, with too heavy
a tail, cannot push the single-symbol error below $10^{-3}$ within $\gamma\le20$.
Thus the single-symbol error rate is set only by the marginal tail of the jitter,
a heavy tail makes $2Q(\gamma/4)$ systematically underestimate, and the floor must
be recalibrated to a larger $c$ from the measured jitter distribution; the
temporal correlation of jitter affects only multi-symbol/frame decoding, not the
single-symbol floor.

\begin{figure*}[t]
\centering
\includegraphics[width=0.98\textwidth]{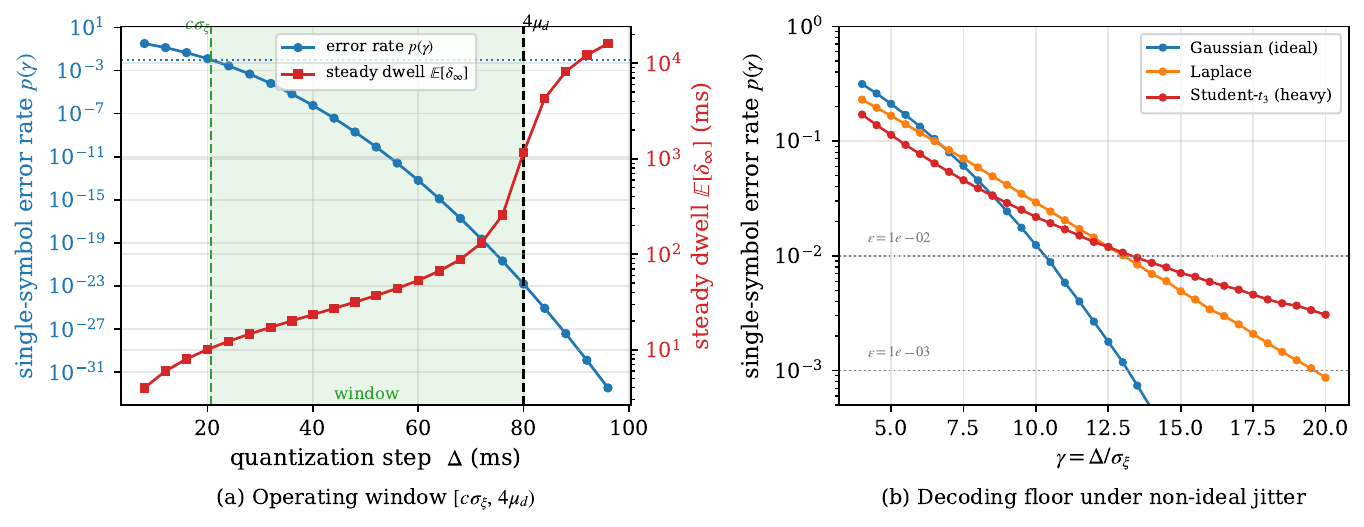}
\caption{The quantization-step operating window and the decoding floor under
non-ideal jitter.}\label{fig:window}
\end{figure*}

\subsection{Real-data experiments}\label{sec:num_realtraffic}

This section checks the criterion on real traffic in two steps: first in the safe
region ($\Delta\in\{8,12,16\}$\,ms) whether the stability margin is positive and
the steady dwell bounded, then sweeping the step along the real background to
approach and cross the boundary $\rho=1$ to see whether the criterion correctly
locates the divergence critical point. In both steps every simulated chain's window
is drawn from within a single candidate flow, never spliced across flows (see the
extraction protocol below).

Real IPDs are taken from two encrypted-traffic datasets released by the Canadian
Institute for Cybersecurity (CIC/UNB): ISCXTor2016~\cite{Lashkari2017Tor} and
ISCXVPN2016~\cite{DraperGil2016VPN}. Both collect per-packet pcaps of browsing,
chat, streaming audio/video, file transfer, VoIP, and mail/P2P in a controlled
environment, covering Tor-anonymized, VPN-tunneled, and non-anonymized (NonTor)
transport; Tor and NonTor come from ISCXTor2016, VPN from ISCXVPN2016.

We extract the per-packet unidirectional in-flow IPD sequence, not flow-level means
or session-level statistics: within each application we recover the packet arrival
times inside a single unidirectional flow, take adjacent-arrival differences as
millisecond IPDs, and restrict to $[1,500]$\,ms to remove sub-millisecond
capture/batch-send effects and over-long session idles. The simulation unit is a
single real flow: each chain's whole window comes from a contiguous segment inside
one real flow, never spliced across flows---the recursion state $\delta_{i-1}$ has
no physical meaning at a flow boundary, and concatenating IPDs of different
connections would inject a spurious correlation at the splice. Accordingly we keep
only unidirectional flows of length at least $L_{\min}=5000$ IPDs after filtering
($5000$ being the maximum window used later), and each chain's random start and
window are confined to one candidate flow, flows sampled by length so each real
IPD sample is chosen with roughly equal probability.

To avoid mixing very different services into one averaged sample, we choose four
representative application-level scenarios: Tor browsing (anonymous web access),
Tor file-transfer (bulk transfer under anonymization), and VPN audio (continuous
media in a tunnel). The fourth would ideally be NonTor web access as an
un-anonymized baseline, but the NonTor web flows passing the $[1,500]$\,ms filter
are mostly very short: of $1496$ candidate flows the longest has only $1102$ IPDs,
none reaching $L_{\min}=5000$, so it cannot support window simulation without
cross-flow splicing; we therefore use the length-sufficient NonTor audio of the
same dataset as the un-anonymized reference channel. Table~\ref{tab:ipd_extract}
lists the candidate and retained flow statistics (Tor2016, VPN2016 abbreviate
ISCXTor2016, ISCXVPN2016).

\begin{table*}[t]
\centering
\caption{Application-level IPD data extracted by flow boundary and used in the
experiments.}\label{tab:ipd_extract}
\begin{tabular}{@{}llrrrr@{}}
\toprule
Scenario & Source & Candidate & Retained & Total IPDs & Flow-length range \\
\midrule
Tor browsing & Tor2016 & 26 & 12 & 248{,}706 & 5{,}923--44{,}269 \\
Tor file-transfer & Tor2016 & 8 & 6 & 445{,}635 & 39{,}661--127{,}968 \\
VPN audio & VPN2016 & 47 & 6 & 1{,}122{,}044 & 184{,}298--193{,}125 \\
NonTor audio & Tor2016 & 95 & 24 & 834{,}591 & 5{,}421--121{,}579 \\
\bottomrule
\end{tabular}
\end{table*}

\emph{Retained} is the number of candidate flows reaching $L_{\min}$; \emph{Total IPDs}
is the concatenated sample count of those flows (with in-flow offset records that
forbid cross-flow windows). The data thus keeps the in-flow burstiness, short-range
correlation, and heavy-tailed marginal of each real flow rather than compressing
them into one mean per flow, and the offset records ensure every simulation window
falls entirely within a single connection.

Table~\ref{tab:dataset} gives the retained-flow count and basic statistics of the
four scenarios: they consist of $12$, $6$, $6$, $24$ retained flows (\emph{Samples}
is the total in-flow IPD count), and all real-data experiments below run on these
flows with each chain strictly within a single flow. Lag-1 autocorrelation is
computed in-flow, dropping cross-boundary adjacent pairs. The means span
$19$--$40$\,ms. VPN audio has the lowest CV ($0.24$) but extreme kurtosis
($1869.5$), i.e.\ a near-constant media interval with a few far-above-mean outlier
subsequences; the other three have CV $1.5$--$1.8$, a more pronounced right-skewed
heavy tail. Lag-1 autocorrelation spans widely, from $0.02$ for NonTor audio
(almost no residual predictable structure) to $0.37$ for Tor browsing, covering
near-memoryless to markedly autocorrelated real structure---a representative sample
for testing the criterion across correlation strengths.

\begin{table*}[t]
\centering
\caption{Basic statistics of the application-level real IPD data.}\label{tab:dataset}
\begin{tabular}{@{}lrrrrrrrr@{}}
\toprule
Scenario & Retained & Samples & Mean & Median & Std & CV & Kurtosis & Lag-1 AC \\
 & & & (ms) & (ms) & (ms) & & & $\hat\rho_1$ \\
\midrule
Tor browsing       & $12$ & $248{,}706$   & $39.61$ & $10.38$ & $70.29$ & $1.77$ & $15.3$ & $0.37$ \\
Tor file-transfer  & $6$  & $445{,}635$   & $23.08$ & $3.16$  & $35.52$ & $1.54$ & $35.7$ & $0.16$ \\
VPN audio          & $6$  & $1{,}122{,}044$ & $19.25$ & $20.23$ & $4.67$  & $0.24$ & $1869.5$ & $0.35$ \\
NonTor audio       & $24$ & $834{,}591$   & $35.65$ & $30.87$ & $23.64$ & $0.66$ & $9.5$  & $0.02$ \\
\bottomrule
\end{tabular}
\end{table*}

The stability-margin experiment takes three representative steps
$\Delta\in\{8,12,16\}$\,ms, directly showing robustness to step choice. The
criterion depends only on the dimensionless $\rho=\Delta/(4\mu_d)$, not the
absolute $\Delta$: these steps give $\rho$ over $0.05$--$0.21$ for the four
scenarios, all inside the ceilings $4\hat\mu_d$ ($77$--$158$\,ms). Larger $\Delta$
(up to crossing the ceiling) is covered by the synthetic experiment
(Section~\ref{sec:num_synthetic}) and the boundary stress test
(Section~\ref{sec:num_boundary}). In the real continuous-trace experiments, each
group's $6000$ chains are spread over all retained flows of the scenario
(Table~\ref{tab:dataset}): each window lies wholly within one flow, with a
length-weighted random start advanced in original order to keep local burstiness
and autocorrelation, $5000$ packets each (first $1000$ a warm-up); the steady mean
dwell $\E[\delta_\infty]$ is estimated to check whether it stays finite. All
computation and plotting scripts are released.

\subsubsection{Stability margins}

These four backgrounds are far from i.i.d.---lag-1 autocorrelation $\hat\rho_1$
from $0.02$ to $0.37$, with in-flow bursty heavy tails
(Table~\ref{tab:dataset})---exactly the stationary-ergodic, burst-correlated
backgrounds Theorem~\ref{thm:robust} characterizes, so this subsection is also a
real-data check of Theorem~\ref{thm:robust}, complementary to the synthetic Markov
check of Section~\ref{sec:num_synthetic}. Table~\ref{tab:real_tail} gives results
at three steps $\Delta\in\{8,12,16\}$\,ms: the intensity $\rho=\Delta/(4\hat\mu_d)$
spans $0.05$--$0.21$, all inside $(0,1)$; the steady mean dwell $\E[\delta_\infty]$
is always finite and increases monotonically with $\Delta$. Thus, under the
burstiness and autocorrelation of real application flows, the criterion robustly
gives the correct engineering direction: as long as the time-average drift keeps
enough margin, the buffer does not diverge from short dense bursts. The stronger
the autocorrelation or right-skew, the faster the dwell grows with $\Delta$: Tor
browsing rises from $4.45$\,ms at $\Delta=8$ to $11.32$\,ms at $\Delta=16$, Tor
file-transfer from $4.89$ to $11.59$\,ms; the near-symmetric VPN audio rises only
from $4.00$ to $8.05$\,ms. This foreshadows Section~\ref{sec:num_boundary}: the
pooled-mean ceiling is optimistic for strongly skewed, dispersed-intensity
services, so configuration should keep a margin below the boundary.

\begin{table*}[t]
\centering
\caption{Stability margins of real continuous traces at several quantization steps
($\Delta\in\{8,12,16\}$\,ms).}\label{tab:real_tail}
\begin{tabular}{@{}lrcccccc@{}}
\toprule
Scenario & Ceiling & \multicolumn{3}{c}{$\rho=\Delta/(4\hat\mu_d)$} & \multicolumn{3}{c}{$\E[\delta_\infty]$ (ms)} \\
\cmidrule(lr){3-5}\cmidrule(lr){6-8}
 & $4\hat\mu_d$ (ms) & $\Delta{=}8$ & $12$ & $16$ & $\Delta{=}8$ & $12$ & $16$ \\
\midrule
Tor browsing       & $158.45$ & $0.050$ & $0.076$ & $0.101$ & $4.45$ & $7.31$ & $11.32$ \\
Tor file-transfer  & $92.31$  & $0.087$ & $0.130$ & $0.173$ & $4.89$ & $8.07$ & $11.59$ \\
VPN audio          & $76.99$  & $0.104$ & $0.156$ & $0.208$ & $4.00$ & $6.03$ & $8.05$  \\
NonTor audio       & $142.62$ & $0.056$ & $0.084$ & $0.112$ & $4.04$ & $6.09$ & $8.18$  \\
\bottomrule
\end{tabular}
\end{table*}

\subsubsection{Stress test approaching the stability boundary}\label{sec:num_boundary}

The four scenarios above ($\Delta\in\{8,12,16\}$\,ms) have $\rho$ in the safe
region $0.05$--$0.21$, far from $\rho=1$, so they only confirm \emph{stable far from
the boundary} and cannot test whether the criterion truly holds at the critical
point. To fill this side, for each real scenario we sweep the step along
$\rho=\Delta/(4\hat\mu_d)$ from the safe region to approach and cross $1$ (taking
$\Delta$ as a series of multiples of $4\hat\mu_d$), estimating the steady mean
dwell $\E[\delta_\infty]$ on the retained-flow-bounded continuous traces.

Figure~\ref{fig:boundary}(a) shows the four scenarios' $\E[\delta_\infty]$ rising
sharply as $\rho$ approaches $1$ and amplifying further beyond $1$. But a single
window's mean cannot distinguish a bounded-but-large steady state from unbounded
growth: for the weakly right-skewed VPN audio (median $\approx$ mean),
$\E[\delta_\infty]$ is still about $60$\,ms at $\rho=0.8$, jumping to $663$ then
$1678$\,ms at $\rho=0.95,1.0$, the knee cleanly at $\rho=1$. To separate \emph{large
but bounded} from \emph{divergent}, Figure~\ref{fig:boundary}(b) gives the
window-length diagnostic: on VPN audio, fixing $\rho$ and growing the window $n$
from $2000$ to $40000$ packets---at $\rho=0.9$, $\E[\delta_\infty]$ stays at
$152$--$167$\,ms (nearly window-invariant, growth about $1.07\times$), a bounded
steady state; at $\rho=1.0,1.1$ it grows approximately linearly with the window
($812\!\to\!11564$, $2546\!\to\!46845$\,ms, about $14$--$18\times$), i.e.\
unbounded divergence. The criterion $\rho<1$ thus correctly locates the divergence
critical point on real backgrounds, not only confirming stability far away.

However, on strongly in-flow right-skewed scenarios with dispersed cross-flow
intensity (Tor browsing, Tor file-transfer, NonTor audio), the ceiling from the
overall mean $\hat\mu_d$ is optimistic. Their window diagnostics show that even at
pooled $\rho=0.9$, nominally stable, $\E[\delta_\infty]$ already grows with the
window; the cause is a sizeable dispersion of per-flow means, so the pooled mean
transfers the idle slack of low-intensity flows to high-intensity ones, whereas a
real embedder cannot drain dwell across flows. This does not contradict
Theorem~\ref{thm:robust}; it shows the \emph{time-average} must be that inside a
single real queue, not one obtained by pooling many disconnected flows after the
fact. Hence a direct configuration rule: judge stability and configure by feeding
the densest (smallest-mean) flow's or a low-quantile per-flow mean, not the pooled
mean, into $\Delta<4\bar\mu_d$---keeping a margin below the pooled boundary
commensurate with the cross-flow intensity dispersion.

Finally, the physical meaning of these numbers. In Figure~\ref{fig:boundary}, the
steady dwell at $\rho\ge1$ reaches thousands to tens of thousands of milliseconds
(at $\rho=1.1$, VPN audio accumulates $4.68\times10^4$\,ms $\approx47$\,s on the
longest window); this is simulated on the budget-free unconstrained recursion
(Eq.~\eqref{eq:delta_update}) and characterizes how large the dwell would grow
without a cap, cleanly separating divergence from boundedness---not a delay
actually observed in deployment. Such dwell never appears in a real system:
tens-of-seconds delay far exceeds transport/application timeouts, so the carried
connection resets or drops first; and the embedder's finite budget $\bB$
(Section~\ref{sec:buffer}) truncates dwell to an overflow event (the output IPD
loses its anchor) once it touches $\bB$. Thus the huge $\E[\delta_\infty]$ at $\rho\ge1$ should be read as the
numerical fingerprint of a diverging unconstrained dwell, whose deployment
meaning is that any finite budget is breached persistently.

\begin{figure*}[t]
\centering
\includegraphics[width=0.92\textwidth]{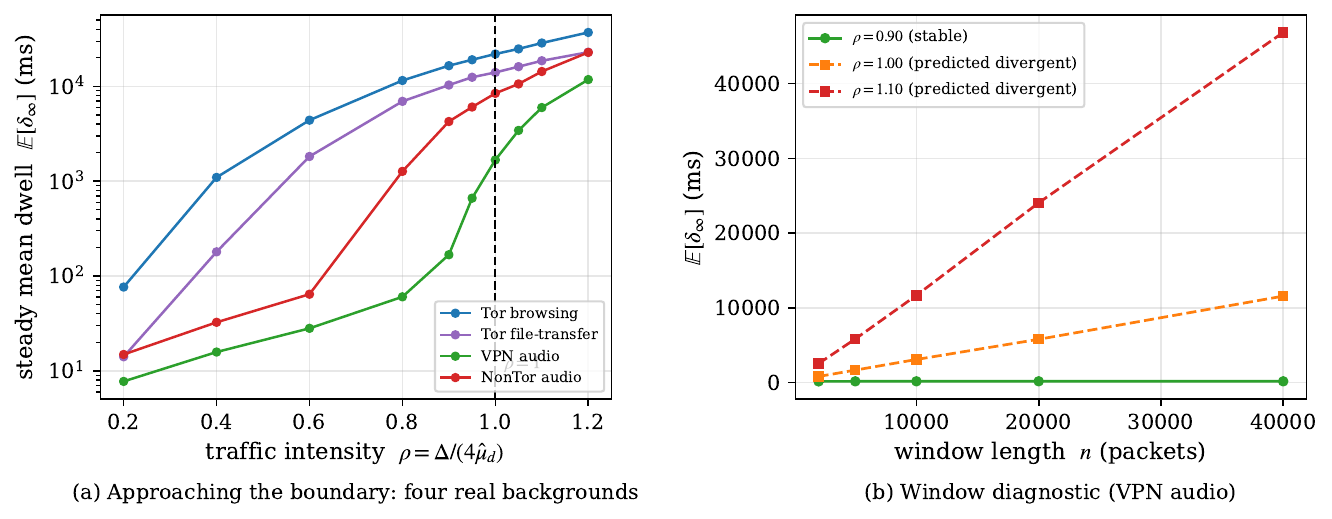}
\caption{Stress test approaching the stability boundary.}\label{fig:boundary}
\end{figure*}

\section{Conclusion}\label{sec:conclusion}

This paper studies causal IPD-QIM watermarking for active attribution of encrypted
network flows, with the core question: for what quantization step $\Delta$ can the
watermark be embedded continuously without the buffer dwell growing without bound?
Lemma~\ref{lem:reduction} writes the causal-QIM dwell recursion in Lindley form,
but this substitution only provides the recursion structure and does not reduce the
fixed-lattice injection to an exogenous uniform input. We therefore analyze
directly the busy-state drift of the fixed dual-lattice, equiprobable-random-bit
rule in the large-dwell region, where the key fact is that once the effective
interval is pressed to zero the mean injection converges to $\Delta/4$.

Based on this drift structure, we obtain a stability criterion away from the
critical boundary: for i.i.d.\ backgrounds the system is stable iff
$\mu_d>\Delta/4$ (i.e.\ $\Delta<4\mu_d$, Theorem~\ref{thm:stability}); for
stationary-ergodic bursty backgrounds it generalizes to $\bar\mu_d>\Delta/4$
(i.e.\ $\bar\rho<1$, Theorem~\ref{thm:robust}). Stability is thus decided by the
time-average drain capacity, and local instantaneous overload does not necessarily
cause divergence. Queue stability gives the upper end $4\bar\mu_d$ and decoding
reliability the exogenous lower end $\Delta\ge c\sigma_\xi$ ($c=4Q^{-1}(\epsilon/2)$);
together they form the conservative operating window
$\Delta\in[c\sigma_\xi,4\bar\mu_d)$, separating \emph{can decode} from \emph{can embed
stably}---the lower end controlled by the channel jitter, the upper end determined
intrinsically by the causal queue dynamics.

Numerical experiments test the criterion at two levels. Synthetic backgrounds that
exactly meet the theoretical assumptions confirm the recursion identity, the busy
injection, the transition at $\rho=1$, the critical point's independence of the
distribution family, and the near-critical divergence scaling; real application-
level IPD experiments, without assuming i.i.d.\ or finite-state Markov structure,
draw continuous segments inside single real flows and avoid spurious draining from
cross-flow splicing. On four real continuous traces, as long as the empirical
intensity is well below the boundary the steady dwell stays finite; as the step
approaches and crosses $\rho=\Delta/(4\hat\mu_d)=1$, the unconstrained dwell shows
clear critical amplification and divergence. The real data also expose a
configuration boundary: the cross-flow pooled mean can overestimate the stability
margin, so stability should be judged by the time-average inside a single queue,
using the densest-flow or a low-quantile per-flow mean when needed.

The applicability of the criterion is bounded by the drift computation itself. The
background must have a well-defined long-run time average, and the embedder must
keep the fixed dual-lattice, equiprobable-random-bit rule analyzed here; once
per-symbol dither, non-equiprobable coding, dynamic phase, or a feedback step is
added, the busy-state mean injection is no longer necessarily $\Delta/4$ and the
stability upper bound must be re-established from the new rule. The budget $\bB$
likewise does not enter the stability threshold: it only measures, after the
unconstrained recursion is already stable, the tail risk that the required
exact-alignment dwell exceeds the available budget; if the time-average drift is
nonnegative, any finite budget can only postpone overflow, not restore stability.
Hence what we provide is not a complete deployment protocol but a stability
pre-screen for fixed-step causal QIM.

Beyond this pre-screen, several harder problems emerge naturally. The closest is an
overflow-and-budget theorem: for $\rho<1$, the decay of the steady tail
$\Prob(\delta_\infty>\bB)$ with the budget is a large-deviations problem, whose
exponential-versus-subexponential dichotomy and the resulting buffer-sizing and
effective-bandwidth characterization decide how rare overflow is given stability.
Complementary is active control in the overloaded regime: when the time-average
drift is nonnegative a fixed step must diverge, so one must shed load online on
detecting overload---adaptive step, selective dropping, or congestion-triggered
injection back-off---and re-establish a stability guarantee in closed loop; this in
turn requires online estimation of the time-average drain $\bar\mu_d$ on
non-stationary, session-drifting backgrounds, pushing the offline criterion toward
real-time control. A more fundamental direction comes from the core quantity
itself: the busy-state mean injection $\E_w[R_w(0)]$ is the sole embedding-side
quantity setting the stability ceiling, and the $\Delta/4$ of the fixed dual
lattice and the $\Delta/2$ under dither are only two of its endpoints; a unified
characterization of this quantity and the corresponding threshold across general
QIM variants (dither, non-equiprobable coding, dynamic phase, multi-lattice) would
lift the criterion from a single rule to a family. Finally, incorporating the real
(correlated, heavy-tailed) channel-jitter distribution, blind frame
synchronization, and soft decoding into a joint design with queue stability is what
turns the two-sided window into an end-to-end decodable deployment system.

\section*{Acknowledgment}
This work was supported by the National Natural Science Foundation of China
(Grants No. U2436601).

\bibliographystyle{IEEEtran}
\bibliography{references}

\end{document}